\documentclass[11pt]{article}
\usepackage[utf8]{inputenc}
\usepackage[margin=1in]{geometry} 
\usepackage[usenames,x11names]{xcolor}

\usepackage{authblk}

\usepackage{amsmath,amssymb,amsthm}
\usepackage{mathtools}
\usepackage{thm-restate}

\usepackage{enumerate} 

\usepackage[linesnumbered,ruled,vlined]{algorithm2e}

\usepackage[colorlinks,linkcolor=DeepSkyBlue3,citecolor=SpringGreen3]{hyperref}
\usepackage[nameinlink,capitalise]{cleveref} 

\usepackage{float}

\usepackage{xspace}

\usepackage{comment}

\newtheorem{theorem}{Theorem}

\newtheorem{lemma}[theorem]{Lemma}

\theoremstyle{definition}
\newtheorem{definition}[theorem]{Definition}
\newtheorem{remark}[theorem]{Remark}
\theoremstyle{plain}

\crefname{question}{Question}{Questions}

\DeclareMathOperator{\poly}{poly}
\DeclareMathOperator{\polylog}{polylog}

\DeclareMathOperator*{\argmax}{arg\,max}

\newcommand{\rr}{r}

\newcommand{\rhoIndepSet}{\ensuremath{\rho}\textnormal{-\textsc{IndepSet}}\xspace}

\title{A Tolerant Independent Set Tester}
\author{Cameron Seth}
\affil{University of Waterloo}

\begin{document}
\maketitle

\begin{abstract}
    We give nearly optimal bounds on the sample complexity of $(\widetilde{\Omega}(\epsilon),\epsilon)$-tolerant testing the $\rho$-independent set property in the dense graph setting.
    In particular, we give an algorithm that inspects a random subgraph on $\widetilde{O}(\rho^3/\epsilon^2)$ vertices and, for some constant $c,$ distinguishes between graphs that have an induced subgraph of size $\rho n$ with fewer than $\frac{\epsilon}{c \log^4(1/\epsilon)} n^2$ edges from graphs for which every induced subgraph of size $\rho n$ has at least $\epsilon n^2$ edges.
    Our sample complexity bound matches, up to logarithmic factors, the recent upper bound by Blais and Seth (2023) for the non-tolerant testing problem, which is known to be optimal for the non-tolerant testing problem based on a lower bound by Feige, Langberg and Schechtman (2004).

    Our main technique is a new graph container lemma for sparse subgraphs instead of independent sets.
    We also show that our new lemma can be used to generalize one of the classic applications of the container method, that of counting independent sets in regular graphs, to counting sparse subgraphs in regular graphs.
\end{abstract}

\section{Introduction}
Graph property testing, first studied by Goldreich, Goldwasser and Ron \cite{goldreichPropertyTesting1998}, is the problem of distinguishing between graphs that have some property $\Pi$ from graphs that are sufficiently \emph{far} from having the property $\Pi$ by inspecting only a small fraction of the graph.
In particular, in the dense graph model, a graph $G$ is said to be $\epsilon$-far from the property $\Pi$ if at least $\epsilon n^2$ edges must be added or removed from $G$ to make it have the property.
An $\epsilon$-tester for graph property $\Pi$ is a bounded error\footnote{By \emph{bounded error} we mean there exist absolute constants $\delta_1>\delta_2$ such that if $G$ has property $\Pi$ then the algorithm accepts with probability at least $\delta_1,$ and if $G$ is $\epsilon$-far from $\Pi$ then the algorithm accepts with probability at most $\delta_2.$} randomized algorithm that distinguishes between graphs that have the property $\Pi$ from graphs that are $\epsilon$-far from $\Pi.$
An $\epsilon$-tester with sample complexity $s$ is an algorithm that takes a random sample $S$ of $s$ vertices and inspects only $G[S]$ of the input graph $G.$
We say that the sample complexity of $\epsilon$-testing $\Pi$ is the minimum $s$ such that there exists a $\epsilon$-tester for $\Pi$ with sample cost $s.$

Tolerant testing, introduced by Parnas, Ron and Rubinfeld \cite{parnas2006tolerant}, is a natural generalization of testing where the goal is to distinguish between graphs that are $\epsilon_1$-close to some property $\Pi$ from graphs that are $\epsilon_2$-far from $\Pi,$ where $\epsilon_1$-close means that at most $\epsilon_1 n^2$ edges need to be added or removed from the graph to make it have the property.
Such a tester is called an $(\epsilon_1,\epsilon_2)$-tolerant tester, and the goal is to find the minimum sample complexity of any such tester, as a function of $\epsilon_1,\epsilon_2$ and the input size.
It is desirable for such a tester to apply for any $\epsilon_1<\epsilon_2,$ referred to as a fully tolerant tester, but this is not necessary.
For example, if $\epsilon_1=0$ then this corresponds to the standard notion of testing.
In general, $\epsilon_1$ may be a function of $\epsilon_2.$

Testing the property of having an independent set\footnote{Testing the property of having a $\rho n$ independent set is equivalent to testing the property of having a $\rho n$ clique by considering the complement of the graph. \cite{goldreichPropertyTesting1998} studied the clique property but for our work it is more convenient to phrase in terms of independent sets.} of size $\rho n,$ denoted \rhoIndepSet, is one of the original problems studied in \cite{goldreichPropertyTesting1998}, where they gave an $\epsilon$-tester with sample complexity $\widetilde{O}(\rho/\epsilon^4).$\footnote{Here and throughout we use $\widetilde{O}(\cdot)$ and $\widetilde{\Omega}(\cdot)$ to hide terms that are polylogarithmic in the argument.}
The standard version is now very well understood.
In particular, recent work by Blais and Seth \cite{blaisSeth} showed that it is possible to $\epsilon$-test the $\rho$-indpendent set property with sample complexity $\widetilde{O}(\rho^3/\epsilon^2)$, and this is optimal, up to logarithmic factors, based on a lower bound by Feige, Langberg and Schechtman \cite{feigeCliqueTesting2004}.

Despite the tight bounds for the sample complexity of the standard version of $\epsilon$-testing the \rhoIndepSet property, little is known for the tolerant version.
Shown in \cite{parnas2006tolerant}, any tester that takes a random sample $S$ of vertices and inspects only $G[S]$ has a very small amount of tolerance by default based on the fact that a random subgraph won't find any violating edges in a graph that is very close to the property,\footnote{They actually show this for testers which make edge queries to the input which are uniformly distributed, but a similar argument applies for testers where the vertex samples are uniformly distributed.} and so the tester of Blais and Seth \cite{blaisSeth} gives a $(\epsilon^4,\epsilon)$-tolerant tester with sample complexity $\widetilde{O}(\rho^3/\epsilon^2).$
However, this is described by \cite{parnas2006tolerant} as ``weak" tolerance, and the real question is: when does there exist an efficient tolerant tester beyond this weak amount of tolerance given by the standard tester?

There have been no direct results analyzing this question for tolerantly testing \rhoIndepSet.
The best prior result comes from a general result for tolerant testing of general graph partition properties, which is a class of properties for which each property corresponds to graphs that can be partitioned into $k$ parts with specific bounds on the edge densities between parts.
This class of properties includes $k$-colorability, \rhoIndepSet (with $k=2$), and many others.
Fiat and Ron show that for any $\epsilon_1 < \epsilon_2$ there is a $(\epsilon_1,\epsilon_2)$-tolerant tester for any graph partition property with sample complexity $\poly(k,1/(\epsilon_2-\epsilon_1))$ \cite[Appendix C]{fiat2021efficient}.
These testers work by making multiple calls to a standard tester for general graph partition properties (\cite{goldreichPropertyTesting1998}, \cite{fischer2010approximate} or \cite{shapira2024GPP}, for example) and as a result have large polynomial dependence on $1/\epsilon.$ 
There is also a general result by Fischer and Newman saying that any graph property that has an $\epsilon$-tester with sample complexity $s(\epsilon)$ has an $(\epsilon_1,\epsilon_2)$-tolerant tester, for any $\epsilon_1<\epsilon_2,$ with sample complexity that is tower type in $s(\epsilon_2-\epsilon_1)$ \cite{fischerNewman2007testing}, and this was recently improved by Shapira, Kushnir and Gishboliner to a doubly exponential bound in $s(\epsilon_2-\epsilon_1)$ \cite{shapira2024revisited}.

\subsection{A Tolerant Independent Set Tester}
Our main contribution is a $(\widetilde{\Omega}(\epsilon),\epsilon)$-tolerant tester for \rhoIndepSet with nearly optimal sample complexity.
We note that our theorem applies for any $\epsilon > 0,$ however the problem is only non-trivial when $\epsilon < \rho^2/2$ because every graph is $\rho^2/2$-close to $\rhoIndepSet.$

\begin{restatable}{theorem}{tolTester}
    \label{thm:tolerantTester}
    There exists a constant $c$ such that for any $\epsilon>0$ the sample complexity of $\left(\frac{\epsilon}{c \log^4(1/\epsilon)},\epsilon\right)$-tolerant testing the $\rhoIndepSet$ property is $\widetilde{O}(\rho^3/\epsilon^2).$
\end{restatable}

The sample complexity in \cref{thm:tolerantTester} matches the sample complexity, up to logarithmic factors, of the non-tolerant tester for \rhoIndepSet by Blais and Seth \cite{blaisSeth}, which itself is known to be optimal, up to logarithmic factors, based on a lower bound on the non-tolerant version by Feige, Langberg and Schechtman \cite{feigeCliqueTesting2004}.
In other words, \cref{thm:tolerantTester} shows that, surprisingly, $(\widetilde{\Omega}(\epsilon),\epsilon)$-tolerant testing \rhoIndepSet is no harder than the non-tolerant version.

In other settings, there are properties that demonstrate that tolerant testing is, in general, much harder than non-tolerant testing.
In particular, in the bounded degree model, Goldreich and Wigderson gave a graph property for which there is a subexponential gap between the query complexity of non-tolerant testing and the query complexity of tolerant testing (i.e.\ the query complexity in the non-tolerant setting is $q$ and the query complexity in the tolerant setting is $\exp(q^c)$ for some constant $c$) \cite{goldreich2022TolerantvsStandard}.
Fischer and Fortnow show a similar separation for testing a property of boolean functions \cite{fischer2005tolerantvsNontolerant}.
It is not known whether such a strong separation exists in the dense graph setting.

We note that it still remains open whether there exists a fully tolerant tester (a $(\epsilon_1,\epsilon_2)$-tolerant tester for any $\epsilon_1$ and $\epsilon_2$) with $\poly(1/(\epsilon_2-\epsilon_1))$ sample complexity.
Efficient fully tolerant testers are desirable due to the fact that they imply an efficient distance approximation algorithm \cite{parnas2006tolerant}.
There are challenges with strengthening our techniques to give a fully tolerant tester, see \cref{rem:logNecessity} for more details.

We also note testing algorithms can also be measured by their edge query complexity instead of sample complexity.
The edge query complexity is the number of individual edges inspected by the algorithm, and such algorithms may not operate by taking a sample of vertices and inspecting all the edges in the sample.
However, it has been shown by Goldreich and Trevisan \cite{GoldreichTrevisan03} that there is at most a quadratic gap between the edge query complexity and the sample complexity of testing any graph property, and a similar argument applies for tolerant testing of graph properties.
Blais and Seth \cite{blaisSeth2024new} show that there is a gap between the edge query complexity of \rhoIndepSet and the number of edges inspected by the optimal sample based tester, but it remains open whether a similar gap exists in the tolerant setting for \rhoIndepSet.

\subsection{Main Technique: New Container Lemma}
Our main technique to prove \cref{thm:tolerantTester} is the graph container method, which is a tool from extremal combinatorics that has seen many applications since its generalization to the hypergraph container method by Balogh, Morris and Samotij \cite{baloghIndependentSetsHypergraphs2015} and Saxton and Thomason \cite{saxtonThomasonHypergraphContainers2015}.
The idea of the container method is that for any graph $G=(V,E)$ that satisfies certain conditions, any independent set $I$ in $G$ has a very small fingerprint $F \subseteq I$, which defines a relatively small container $C(F) \subseteq V$, such that the independent set $I$ is a subset of the container (which we will often phrase as: the independent set $I$ is \emph{contained} in the container).
In a sense, the fingerprint of an independent set $I$ can be viewed as the most ``informative" vertices in the independent set, and the function $C(F)$ corresponds to using this information to construct a set of potential candidates for the vertices of the independent set.
The key observation is that since each fingerprint is small, there can only be so many fingerprints, and since each fingerprint defines a container, then there can only be so many containers.
This method, in the form of a container lemma for independent sets in graphs that are $\epsilon$-far from \rhoIndepSet, was used by Blais and Seth to get the optimal non-tolerant tester for the \rhoIndepSet property \cite{blaisSeth}.

We wish to use this technique to analyze a tolerant tester for \rhoIndepSet, but in order to do so we require a container lemma for sparse subgraphs\footnote{Here and throughout we say ``subgraph" to mean induced subgraph.} instead of independent sets.
The main challenge with proving such a lemma has to do with the standard way that the collection of fingerprints is constructed.
In particular, a fingerprint for an independent set $I$ in $G=(V,E),$ and the associated container that contains $I,$ is constructed through a greedy procedure as follows.
After initiating the container to $V,$ in each step the vertex $v \in I$ with the highest degree in the current container is added to the fingerprint, and the neighbours of $v$ and the vertices with higher degree than $v$ are removed from the current container.
In this way, the independent set is always a subset of the current container.
However, this fails dramatically when dealing with sparse subgraphs: running this procedure on a set $J$ for which $G[J]$ is a sparse subgraph would not guarantee that $J$ remains in the current container in each step because there may be vertices in $J$ that neighbour the selected vertex.
Even if we relax the condition so that we only require that \emph{most} of $J$ is contained in the container, which is actually all that we need to prove \cref{thm:tolerantTester}, the fingerprint procedure still fails because all of $J$ may be removed by a single vertex selected to the fingerprint that neighbours all vertices in $J.$

Due to this challenge, nearly all of the prior work on the graph and hypergraph container method applies to independent sets and not sparse subgraphs.
Container lemmas for sparse subgraphs have only been studied in the work of Saxton and Thomason \cite{saxtonThomasonHypergraphContainers2015}, as well as a very recent work by Nenadov \cite{nenadov2024counting}, however, in both of these works the container lemmas apply only for very sparse subgraphs.
In \cite{saxtonThomasonHypergraphContainers2015} the results apply for subgraphs $G[J]$ for which the number of edges is at most roughly $\frac{1}{|J|} |J|^2$, and in \cite{nenadov2024counting} the results apply for subgraphs $G[J]$ for which the maximum degree is much less than $\frac{\epsilon}{\rho^2}|J|$ (in graphs that are $\epsilon$-far from \rhoIndepSet).
For our application we require a container lemma for a much larger range of sparse subgraphs.
In particular, we require a lemma that applies for sparse subgraphs with edge density up to roughly $\frac{\epsilon}{\rho^2},$ which is the minimum edge density of the entire graph provided it is $\epsilon$-far from \rhoIndepSet, and for sparse subgraphs $J$ that may have maximum degree up to $|J|.$
None of the prior results apply to this setting, and it is not even clear if such a container lemma should exist.

We address these challenges and give a new container lemma for sparse subgraphs in graphs that are $\epsilon$-far from \rhoIndepSet.
In order to state our lemma, we require a new definition of fingerprints.

\begin{definition}
    \label{def:fingeprint}
    Let $U \subseteq V$ be a subset of vertices in a graph $G=(V,E)$.
    A \emph{fingerprint} on $U$ is a pair $(F,R)$ where $F$ is a sequence $(f_1,f_2,\dots,f_{|F|})$ where each $f_i \in U \times \{\uparrow,\downarrow\}$ and $R \in \left([|F|] \times U \right) \cup \{\bot\}.$
    $F$ is called the fingerprint sequence, and $R$ is called the revision.
    The $\bot$ is used to denote that there is no revision.
    We use $\mathcal{F}(U)$ to denote the set of all possible fingerprints on $U.$
\end{definition}

We can now state our new graph container lemma for sparse subgraphs.

\begin{restatable}{lemma}{lemGCL}
    \label{lem:GCL}
    Let $\epsilon < \rho^2/2$ and let $\ell \geq 1024 \log^4(1/\epsilon).$
    If $G=(V,E)$ is $\epsilon$-far from having a $\rho n$ independent set then there exists a set of fingerprints $\mathcal{F} \subseteq \mathcal{F}(V)$ and a function $C: \mathcal{F} \rightarrow P(V)$ that satisfies the following:\footnote{$P(V)$ denotes the power set of $V.$}
    \begin{enumerate}
        \item  For any $(F,R) \in \mathcal{F},$ $|C(F,R)| \leq \left(1-\frac{\epsilon}{2\rho^2}\right)\rho n.$
        Further, letting $\alpha_{(F,R)}$ denote the value in $[\frac{\epsilon}{2\rho^2},1]$ that satisfies $|C(F,R)| = \left(1-\alpha_{(F,R)}\right)\rho n,$ then
        \item For any $(F,R) \in \mathcal{F},$ $|F| \leq \frac{32 \alpha_{(F,R)} \rho^2 \log^2(8\rho/\epsilon)}{\epsilon},$ and
        \item For any $J \subseteq V$ for which $G[J]$ has fewer than $\frac{\epsilon}{\ell \rho^2}|J|^2$ edges, there exists $(F,R) \in \mathcal{F}(J) \cap \mathcal{F}$ such that $|J \cap C(F,R)| \geq \left(1-\frac{32\alpha_{(F,R)} \log^2(8\rho/\epsilon)}{\sqrt{\ell}}\right)|J|.$
    \end{enumerate}
\end{restatable}

To prove \cref{lem:GCL} we modify the standard container approach in three ways, which we describe briefly below (a more detailed overview can be found in \cref{sect:overview}).
First, as mentioned above, it is useful to change the definition of a fingerprint to decouple the operations of removing the neighbours of the selected vertex and removing the higher degree vertices of the selected vertex.
This is represented by the $\uparrow$ or $\downarrow$ that is paired with each vertex in the fingerprint sequence, where the $\uparrow$ corresponds to the operation of ``remove higher degree vertices", and the $\downarrow$ corresponds to the operation of ``remove neighbours".
In other words, instead of viewing a fingerprint as a set of the most informative vertices, we view a fingerprint as a sequence of informative vertices along with instructions of $\uparrow$ or $\downarrow$ to denote how to construct the container using each fingerprint vertex.
The second change has to do with how fingerprint vertices are selected.
In the case of constructing a fingerprint for an independent set, at each step the most ``informative" vertex is the vertex that tells the most information about the independent set by revealing lots of vertices that are \emph{not} in the independent set.
For a sparse subgraph $J,$ the most informative vertex is the vertex that can eliminate many vertices from the current set of candidates, without eliminating too many vertices from $J,$ and so the greedy procedure now selects the vertex that maximizes this ratio.
The final change is due to the fact that, even with the above modifications, the fingerprint/container procedure still may not reach a container that contains a sufficiently high fraction of the sparse subgraph. 
As a result, we do an additional step that makes a ``revision" to the container.
The revision is represented by an index of the fingerprint sequence and a vertex in the sparse subgraph, and can be viewed as an additional instruction on how to revise the container to meet the desired containment bounds.

We note that there is an important trade off in \cref{lem:GCL} between the size of the container and the size of the fingerprint, and this trade off is balanced with the fraction of $J$ that is contained in the container.
In particular, for sufficiently large $\ell$ (at least roughly $2056 \log^4(1/\epsilon)$), for any sparse subgraph $G[J]$ there exists a fingerprint $(F,R) \in \mathcal{F}$ such that $|C(F,R)|=(1-\alpha)\rho n$ and at least $(1-\alpha/2)|J|$ of $J$ is contained in the container.
These trade offs are necessary in order to prove the optimal bounds on the sample complexity in \cref{thm:tolerantTester}.

\subsection{Counting Sparse Subgraphs in Regular Graphs}
Our new graph container lemma generalizes the standard container approach to handle sparse subgraphs instead of independent sets.
While our intended application was the tolerant tester of \cref{thm:tolerantTester}, we expect it should have other applications.
We demonstrate this by applying \cref{lem:GCL} to generalize one of the classic applications of the container method, that of counting independent sets in regular graphs, to counting sparse subgraphs in regular graphs.

The problem of counting independent sets in $d$-regular graphs is very well understood.
Proving a conjecture of Granville, Alon first showed that the number of independent sets in a $d$-regular graph is at most $2^{\frac{n}{2}\left(1+O(d^{-0.1})\right)},$ and conjectured that the graph with $\frac{n}{2d}$ disjoint copies of $K_{d,d},$ which has $\left(2^{d+1}-1\right)^{\frac{n}{2d}}$ independent sets, has the maximum number of independent sets \cite{alon1991independent}.
Sapozhenko improved the upper bound by Alon to $2^{\frac{n}{2}\left(1+O(\sqrt{\log(d)/d}))\right)}$ using the graph container method \cite{sapozhenko2001number}.
Since then, Alon's conjecture has been been proven, first by Kahn for bipartite regular graphs \cite{kahn2001entropy}, and then by Zhao for general graphs \cite{zhao2010number}.

But what can be said about the number of sparse subgraphs (of at most some edge density) in a $d$-regular graph?
As far as we know, there have been no results studying this natural generalization of the problem.
The container lemmas for sparse subgraphs by Saxton and Thomason \cite{saxtonThomasonHypergraphContainers2015} and Nenadov \cite{nenadov2024counting} can be used here, but only apply for sparse subgraphs $G[J]$ with very small edge density (roughly $\frac{1}{|J|},$ which is only $O(\frac{1}{n})$ for subgraphs on $\Omega(n)$ vertices) or sparse subgraphs with very small maximum degree (roughly $\frac{d}{n}|J|$).
Using a similar argument to Sapozhenko \cite{sapozhenko2001number}, but using \cref{lem:GCL} instead of a standard graph container lemma for independent sets, we prove the following theorem, which applies for a much larger range of sparse subgraphs.

\begin{restatable}{theorem}{thmCountingSparseSubgraphs}
    \label{thm:countingSparseSubgraphs}
    Let $G=(V,E)$ be a $d$-regular graph.
    There exists a constant $c$ such that for any $k \geq c \log^9(n)$ the number of subsets $J\subseteq V$ such that $G[J]$ has edge density at most $\frac{1}{k}\frac{d}{n}$ is at most \[2^{\frac{n}{2}\left(1+O\left(\frac{\log^3n}{d}\right)+O\left(\frac{\log^{3}n}{k^{1/3}}\right)\right)}.\]
\end{restatable}

A $d$-regular graph has edge density roughly $\frac{2d}{n},$ and using a Markov bound it is not hard to show that any $d$-regular graph has at least $\frac{1}{2}2^n$ induced subgraphs with edge density less than $\frac{4d}{n}.$
In other words, on one extreme Zhao's result \cite{zhao2010number} says that any $d$-regular graph has fewer than roughly $2^{n/2}$ induced subgraphs with density $0$ (i.e. independent sets), and on the other extreme we have that any $d$-regular graph has at least $\frac{1}{2}2^n$ induced subgraphs with edge density less than $\frac{4d}{n}.$
\cref{thm:tolerantTester} says that for $k$ and $d$ at least $\polylog(n),$ the number of sparse subgraphs with edge density $\frac{1}{k}\frac{d}{n}$ drops to nearly $2^{n/2}.$

Also, it can be shown via a simple counting argument that the graph with $\frac{n}{2d}$ disjoint copies of $K_{d,d}$ has at least $2^{\frac{n}{2}\left(1+\frac{1}{2d}+\Omega\left(\frac{\log(k)}{k}\right)\right)}$ subgraphs with density at most $\frac{1}{k} \frac{d}{n}$ (see \cref{rem:countSparseSubgraphsExample} for details).
This lower bound demonstrates that there is some $\frac{n}{\poly(k)}$ dependence in the exponent, but it remains open what the exact bound should be, and whether the graph with $\frac{n}{2d}$ disjoint copies of $K_{d,d}$ is the extremal graph.

\section{Preliminaries}
\paragraph{Notations.}

Let $G=(V,E)$ be a graph.
Throughout we use $\deg(v)$ to denote the degree of $v$ in $G,$ and for $S \subseteq V$ we use $\deg_{G[S]}(v)$ to denote the degree of $v$ in $G[S].$
For $S,T \subseteq V$ use $E(S)$ to denote the set of edges in $G$ with both endpoints in $S,$ and use $E(S,T)$ to denote the set of edges in $G$ with one endpoint in $S$ and one endpoint in $T.$
We use the notation $S_{\uparrow v}$ to denote all vertices in $S$ with more neighbours in $S$ than $v.$
In other words, \[S_{\uparrow v} = \{u \in S : |N(u) \cap S| > |N(v) \cap S|\}.\]
We define the edge density of a subgraph $S \subseteq V$ as $\frac{|E(S)|}{\binom{|S|}{2}}.$
Throughout, we say \emph{subgraph} to refer to the \emph{induced subgraph} $G[S].$
We use $\log(\cdot)$ to denote $\log_2(\cdot).$
Finally, for a set $S$ we use $P(S)$ to denote the power set of $S.$

\paragraph{Chernoff Bound.}

In order to prove \cref{thm:tolerantTester}, we also use the following form of Chernoff's bound for hypergeometric distributions (see, for example, \cite{blaisSeth}).
\begin{lemma}[Chernoff's Bound]
\label{lem:chernoff}
Let $X$ be drawn from the hypergeometric distribution $H(N,K,n)$ where $n$ elements are drawn without replacement from a population of size $N$ where $K$ elements are marked and $X$ represents the number of marked elements that are drawn.
Then for any $\vartheta \geq E[X]$,
\[
\Pr\big[X \geq \vartheta\big] \leq \exp\left(-\frac{(\vartheta-E[X])^2}{\vartheta+E[X]}\right).
\]
\end{lemma}

\section{A Tolerant Tester for the Independent Set Property}
In this section we prove \cref{thm:tolerantTester} using \cref{lem:GCL}.
The argument is roughly as follows.
The tester takes a random sample $S$ of $s \sim \frac{\rho^3 \polylog(1/\epsilon)}{\epsilon^2}$ vertices and accepts if and only if $S$ contains a set $J$ of size $\rho s$ such that $G[J]$ has at most roughly $\frac{\epsilon}{\polylog(1/\epsilon)}s^2$ edges.
The challenging part is showing that the algorithm rejects instances that are $\epsilon$-far from \rhoIndepSet, and for this we use \cref{lem:GCL}.
By \cref{lem:GCL}, if $G$ is $\epsilon$-far then there is a relatively small collection of fingerprints, and associated containers, such such that for any sparse set $J$ of size $\rho s,$ there exists a fingerprint whose vertices all come from $J,$ such that $J$ is mostly contained in the container associated with the fingerprint.
In particular, as long as we select the parameter $\ell$ in \cref{lem:GCL} to be sufficiently large, the amount of $J$ that is contained in the container is strictly more than the expected number of vertices that would be drawn from the container in a random sample $S.$
Using a Chernoff bound we can upper bound the probability that $S$ contains a sufficiently high number of vertices from any one container, then taking a union bound over all possible fingerprints, and associated containers, gives an upper bound on the probability that $S$ contains a sparse subgraph on $\rho s$ vertices.

We note that the proof of \cref{thm:tolerantTester} follows a similar union bound argument used by \cite{blaisSeth}, but using \cref{lem:GCL} instead of a graph container lemma for independent sets.
The main difference is that instead of upper bounding the probability of finding $\rho s$ vertices from a container of size less than $\rho n,$ we upper bound the probability of finding at least $(1-\alpha/2)\rho s$ vertices from a container of size $(1-\alpha)\rho n.$

\tolTester*
\begin{proof}
    Let $c_1$ be a sufficiently large constant.
    On input $G,$ we want to distinguish between the case that $G$ has a set $U \subseteq V$ of size $\rho n$ such that $G[U]$ has fewer than $\frac{\epsilon}{c_1 \log^{4}(1/\epsilon)} n^2$ edges, from the case that for every $\rho n$ set $U \subseteq V,$ $G[U]$ has at least $\epsilon n^2$ edges.
    To do this, the tester takes a uniformly random set $S$ of $s = c_2 \frac{\rho^3 \log^5(1/\epsilon)}{\epsilon^2}$ vertices without replacement, inspects all possible edges in $G[S],$ and accepts if and only if there is a set $J \subset S$ of size $\rho s$ such that $G[J]$ has at most $2\frac{\epsilon}{c_1 \log^4(1/\epsilon)}s^2$ edges, where $c_2$ is a sufficiently large constant.\footnote{Here and throughout we ignore integer rounding issues as they do not affect the asymptotics of the results.}

    First, we wish to show that the algorithm accepts with probability at least $1/4$ when $G$ contains a sparse subgraph of size $\rho n.$
    In this case, let $U \subseteq V$ be a set of size $\rho n$ such that $G[U]$ has fewer than $\frac{\epsilon}{c_1 \log^{4}(1/\epsilon)} n^2$ edges.
    The probability that $S$ contains $\rho s$ vertices $J$ that form a sparse subgraph $G[J]$ is at least the probability that $S$ contains $\rho s$ vertices $J$ from $U$ that form a sparse subgraph.
    First observe that with probability at least $\frac12,$ $S$ contains at least $\rho s$ vertices from $U$ since the number of vertices drawn from $U$ follows a hypergeometric distribution, and the median of the distribution is at least $\rho s$~\cite{Neumann70,KaasBurhman1980}.
    Next, observe that if $J$ is a random subset of $U$ of size $\rho s,$ then the expected number of edges in $G[J]$ is at most $\frac{\epsilon}{c_1 \log^{4}(1/\epsilon)} n^2 \frac{|J|^2}{\rho^2 n^2}=\frac{\epsilon}{c_1 \log^{4}(1/\epsilon)} s^2.$
    Hence, using Markov's Inequality, for such a random subset $J \subseteq U$ the probability that $G[J]$ has more than $2 \cdot \frac{\epsilon}{c_1 \log^{4}(1/\epsilon)} s^2$ edges is at most $1/2.$
    In other words, the probability that $S$ contains $\rho s$ vertices $J$ from $U$ for which $G[J]$ has at most $2 \cdot \frac{\epsilon}{c_1 \log^{4}(1/\epsilon)} s^2$ edges, given that it contains at least $\rho s$ vertices from $U,$ is at least $\frac12.$
    Combining, we find that the probability that $S$ contains at least $\rho s$ vertices from $U$ \emph{and} $S$ contains $\rho s$ vertices $J$ from $U$ for which $G[J]$ has at most $\frac{\epsilon}{c_1 \log^{4}(1/\epsilon)} s^2$ edges, is at least $\frac12 \cdot \frac12 = \frac14.$

    Now, consider the case that $G$ is $\epsilon$-far from having a $\rho n$ independent set.
    We wish to show that the algorithm accepts with only very small probability.
    First, let $\ell = \frac{c_1 \log^4(1/\epsilon)}{2}.$ 
    Since $c_1$ is a sufficiently large constant, then, by \cref{lem:GCL}, there exists an absolute constant $c_3,$ a collection $\mathcal{F}$ of fingerprints, and a container function $C: \mathcal{F}\rightarrow P(V)$ such that for any $J \subseteq V$ of size $\rho s$ for which $G[J]$ has fewer than $\frac{\epsilon}{\ell \rho^2}|J|^2 = \frac{\epsilon}{\ell}s^2$ edges there exists $(F,R) \in \mathcal{F}(J) \cap \mathcal{F}$ and $\alpha$ such that $|C(F,R)| = (1-\alpha)\rho n,$ $|F| \leq \frac{c_3 \alpha \rho^2 \log^2(1/\epsilon)}{\epsilon}$ and $|J \cap C(F,R)| \geq \left(1-\alpha/2\right)|J|.$
    This implies that for any such sparse subgraph $G[J],$ there exists $(F,R) \in \mathcal{F}(J) \cap \mathcal{F}$ such that
    \[\frac{|J \cap C(F,R)|}{|J|} - \frac{|C(F,R)|}{\rho n} \geq \frac{\epsilon}{2c_3 \rho^2 \log^2(1/\epsilon)} \cdot |F|.\]
    
    Hence, the probability that $S$ contains such a sparse subgraph $G[J]$ is at most
    \begin{equation}
        \label{eqn:tolTesterUpperBound1}
        \sum_{(F,R) \in \mathcal{F}} \Pr\left[(F,R) \in \mathcal{F}(S) \textrm{ and } |S \cap C(F,R)| \geq \left(\frac{\epsilon}{2c_3 \rho^2 \log^2(1/\epsilon)} \cdot |F| + \frac{|C(F,R)|}{\rho n}\right) \rho s\right].
    \end{equation}

    Now, for any $(F,R) \in \mathcal{F},$ let $F_{\rm vert}$ be the set of vertices appearing in the sequence $F$ and the revision $R,$ and observe that $(F,R) \in \mathcal{F}(S)$ implies that $F_{\rm vert} \subseteq S.$
    Then, by breaking up the sum in \cref{eqn:tolTesterUpperBound1} by fingerprints $(F,R)$ for which the size of $F_{\rm vert}$ is $f=1,\dots,f_{\max} = \frac{c_3 \rho^2 \log^2(1/\epsilon)}{\epsilon}+1,$ we can upper bound the probability by 

    \[\sum_{f=1}^{f_{\max}} \sum_{\substack{(F,R) \in \mathcal{F} \\ |F_{\rm vert}| = f}} \Pr[F_{\rm vert} \subseteq S] \cdot \Pr\left[|S \cap C(F,R)| \geq \left(\frac{\epsilon}{2c_3 \rho^2 \log^2(1/\epsilon)} \cdot |F| + \frac{|C(F,R)|}{\rho n}\right) \rho s ~\Big|~ F_{\rm vert} \subseteq S\right].\]

    Now, for any $(F,R),$ if $|F_{\rm vert}|=f$ then \[\Pr[F_{\rm vert} \subseteq S]=\frac{\binom{n-f}{s-f}}{\binom{n}{s}} = \frac{\binom{s}{f}}{\binom{n}{f}},\]
    because there are at most $\binom{n}{s}$ ways to choose $S,$ and $\binom{n-f}{s-f}$ ways to choose an $S$ that includes $F_{\rm vert}.$
    Next, we can upper bound the probability that $S$ contains at least $\left(\frac{\epsilon}{2c_3 \rho^2 \log^2(1/\epsilon)} \cdot |F| + \frac{|C(F,R)|}{\rho n}\right) \rho s$ vertices from $C(F,R),$ conditioned on the fact that $S$ contains $F_{\rm vert},$ using a Chernoff bound.
    In particular, let $X$ be a random variable to denote the number of vertices among the other $s-f$ sampled vertices that belong to $C(F,R)$, and so $E[X] = \frac{|C(F,R)|}{n} (s-f) < \frac{|C(F,R)|}{\rho n}\rho s.$   
    By the Chernoff bound in \Cref{lem:chernoff}, the probability that $S$ contains at least $\left(\frac{\epsilon}{2c_3 \rho^2 \log^2(1/\epsilon)} \cdot |F| + \frac{|C(F,R)|}{\rho n}\right) \rho s - f$ vertices from $C(F,R)$ in the other $s - f$ vertices drawn to form $S$ is
    \begin{align*}
        \Pr\left[ X \ge \left(\frac{\epsilon}{2c_3 \rho^2 \log^2(1/\epsilon)} \cdot |F| + \frac{|C(F,R)|}{\rho n}\right) \rho s - f\right]
        & \le \exp\left(\frac{-\left(\frac{\epsilon}{2c_3 \rho^2 \log^2(1/\epsilon)} \cdot |F| \cdot \rho s - f\right)^2}{2\rho s}\right)\\ 
        & \leq \exp\left(-\frac{\epsilon^2 |F|^2 s}{32 c_3^2 \rho^3 \log^4(1/\epsilon)} \right),
    \end{align*}
    where the last inequality uses that $f \leq |F|+1$ and $\frac{\epsilon}{4c_3\rho^2 \log^2(1/\epsilon)} \rho s \geq 1$ because $s$ is sufficiently large.

    Substituting back into the probability bound above, the probability that $S$ contains a $\rho s$ set $J$ for which $G[J]$ has fewer than $\frac{\epsilon}{\ell}s^2$ edges is at most 
    \[
        \sum_{f=1}^{f_{\max}} \sum_{\substack{(F,R) \in \mathcal{F} \\ |F_{\rm vert}| = f}} \frac{\binom{s}{f}}{\binom{n}{f}} ~ \exp\left(-\frac{\epsilon^2 |F|^2 s}{32 c_3^2 \rho^3 \log^4(1/\epsilon)} \right)
        \leq  \sum_{f=1}^{f_{\max}}\frac{1}{\binom{n}{f}}\sum_{\substack{(F,R) \in \mathcal{F} \\ |F_{\rm vert}| = f}}  \exp\left(f \log s -\frac{\epsilon^2 |F|^2 s}{32 c_3^2 \rho^3 \log^4(1/\epsilon)} \right),
    \]
    where the inequality uses the upper bound $\binom{s}{f} \leq s^f \leq \exp (f \log(s)).$
    Hence, since $s = c_2 \frac{\rho^3 \log^5(1/\epsilon)}{\epsilon^2}$ for a sufficiently large constant $c_2,$ and $f \leq |F|+1,$ this can be upper bounded by
    \[
        \sum_{f=1}^{f_{\max}}\frac{1}{\binom{n}{f}}\sum_{\substack{(F,R) \in \mathcal{F} \\ |F_{\rm vert}| = f}}  \exp\left(-4|F| \log(s)\right).
    \]

    We again split the sum up into fingerprints $(F,R)$ such that $|F|=k,$ for $k=f, \cdots, f_{\max},$ and use that there are at most $\binom{n}{f} 2^k f^k \cdot f \cdot k \leq \binom{n}{f} \exp(2k\log(k))$ possible $(F,R)$ pairs such that $|F_{\rm vert}|=f$ and $|F|=k$ to upper bound the probability by 
    \[
        \sum_{f=1}^{f_{\max}}\frac{1}{\binom{n}{f}}\sum_{k=f}^{f_{\max}}\sum_{\substack{(F,R) \in \mathcal{F} \\ |F_{\rm vert}| = f \\ |F| = k}}  \exp\left(-4 k \log(s)\right) \leq \sum_{f=1}^{f_{\max}}\sum_{k=f}^{f_{\max}} \exp\left(2k \log(k) -4k \log(s)\right) \leq \sum_{f=1}^{f_{\max}}\sum_{k=f}^{f_{\max}} \exp\left(-2k \log(s)\right).
    \]
    
    Finally, using that the sum is upper bounded by the case when $k=1,$ we find that the probability that $S$ contains a $\rho s$ set $J$ for which $G[J]$ has fewer than $\frac{\epsilon}{\ell}s^2$ edges is at most
    \[
        f_{\max}^2 \exp\left(-2 \log(s)\right) \leq \exp(-\log(s)) \leq \frac{1}{s} \ll \frac{1}{8},
    \]
    where the first inequality uses that $f_{\max}^2 \leq s.$
\end{proof}

\begin{remark}
    The proof of \cref{thm:tolerantTester} shows that the tester accepts graphs that are $\widetilde{\Omega}(\epsilon)$-close to \rhoIndepSet with probability at least $1/4$ and accepts graphs that are $\epsilon$-far from \rhoIndepSet with probability at most $1/8.$
    With standard error reduction techniques it is possible to boost the probabilities so that the algorithm answers correctly with probability at least $2/3$ without changing the asymptotic sample complexity.  
\end{remark}

\section{Graph Container Lemma for Sparse Subgraphs}
In this section we prove \cref{lem:GCL}.
First we give a proof overview.
Then, in \cref{sect:newProcedures} we give the new fingerprint and container generating procedures.
In \cref{sect:initialContainment,sect:revisionLemma} we prove a number of lemmas that we need to prove \cref{lem:GCL}, before finally proving \cref{lem:GCL} in \cref{sect:proofGCL}.

\subsection{Proof Overview}
\label{sect:overview}
For a graph $G=(V,E)$ that is $\epsilon$-far from \rhoIndepSet, we prove \cref{lem:GCL} by constructing a set $\mathcal{F} \subseteq \mathcal{F}(V)$ of fingerprints, and showing that the collection satisfies the desired conclusions.
In particular, we give a container function $C:\mathcal{F}(V) \rightarrow P(V),$ and for every set $J$ such that $G[J]$ is sparse we describe how to construct a fingerprint $(F,R)$ such that $C(F,R)$ satisfies the conclusions of the lemma: $|C(F,R)|$ is not too large, $|F|$ is not too large, and $C(F,R)$ contains a relatively large amount of $J.$
Then the collection $\mathcal{F}$ is is the set of all fingerprints formed in this way.
In what follows we briefly describe the construction of $F$ and $R,$ what the container function is, and the key ideas towards showing that the constructions satisfy the desired conclusions of the lemma.
Before doing that, first recall that we define a fingerprint as a pair $(F,R)$ such that $F$ is a sequence of vertices along with instructions from $\{\uparrow,\downarrow\},$ and $R$ is a pair $(i,v)$ where $i \in [|F|]$ and $v$ is a vertex (\cref{def:fingeprint}).
We refer to $F$ as the fingerprint sequence and $R$ as the revision.

Let $\ell = \polylog(1/\epsilon).$
For any subgraph $G[J]$ with edge density at most $\frac{\epsilon}{\ell \rho^2}$ (note that if $G$ is $\epsilon$-far from \rhoIndepSet then the edge density of $G$ is at least $\frac{\epsilon}{\rho^2}$), we construct $F$ through an iterative greedy procedure (\cref{alg:Fingerprint}).
First, initialize a container to $V,$ then, at each step, select a vertex $v\in J$ and an operation from $\uparrow$ or $\downarrow,$ where $\uparrow$ corresponds to removing higher degree vertices than $v$ and $\downarrow$ corresponds to removing neighbours of $v$, such that running the operation on the current container removes a large number of vertices in the container and removes a relatively small number of vertices from $J$ in the process.
In particular, in each step the procedure selects the vertex and operation that \emph{maximizes} the ratio of vertices removed from the container divided by the number of vertices of $J$ removed from the container.
We add the pair containing the vertex $v$ and the operation from $\{\uparrow,\downarrow\}$ to the fingerprint sequence $F$, run the operation on the container, and repeat this process until we reach a container with fewer than $\epsilon n^2/4$ edges.
This final container, which for now we call $C_{F},$ can be viewed as a first ``guess" at a container for $J.$

Our first lemma, \cref{lem:weakContainment}, shows that $C_{F}$ already satisfies some useful properties.
In particular, we show that at each step of the procedure described above, the ratio of vertices removed from the current container divided by the vertices removed from $J$ is at least roughly $\sqrt{\ell} \cdot \frac{n}{|J|}.$\footnote{Note that in the full proof this ratio depends on the size of the current container instead of $n,$ but for simplicity here we assume that the current container is of size roughly $n.$}
In other words, if $\frac{1}{\sqrt{\ell}}|J|$ vertices of $J$ are removed from the container, then roughly $n$ vertices are removed from the container, and since the initial container size is of size $n,$ we have that $C_{F}$ contains at least roughly $(1-\frac{1}{\sqrt{\ell}})|J|$ of $J.$

To develop some intuition on why this is the case consider the following example.
Let $G=G(n,\frac{4\epsilon}{\rho^2})$ be the random graph on $n$ vertices where each edge occurs with probability $\frac{4\epsilon}{\rho^2}.$
With high probability this graph is $\epsilon$-far from having a $\rho n$ independent set.
Let $J$ be a subset of vertices such that the edge density of $G[J]$ is at most $\frac{\epsilon}{\ell \rho^2}.$
Consider the first step of the procedure described above.
Observe that the average degree in $G[J]$ is less than $\frac{\epsilon}{\ell \rho^2}|J|,$ hence, the procedure can select a vertex $v \in J$ with at most the average degree in $G[J]$, and the operation $\downarrow$ to remove neighbours, and doing so removes at least roughly $\frac{4\epsilon}{\rho^2}n$ vertices from the container while removing fewer than $\frac{\epsilon}{\ell \rho^2}|J|$ vertices of $J.$
This gives a removal ratio of roughly $4 \ell \cdot \frac{n}{|J|}.$

In order to show the $\sqrt{\ell} \cdot \frac{n}{|J|}$ ratio in general, we require a general lemma (\cref{lem:degreeThresholdExistence}) about the degrees of vertices in a dense graph.
This lemma shows that for any subset $U \subseteq V$ there exists a degree threshold $d$ such that at least roughly $|E(U)|/d$ vertices in $U$ have degree larger than $d$ in $U.$
Hence, in each step of the procedure, at least roughly $\frac{1}{d}\frac{\epsilon}{\rho^2}n^2$ vertices in the current container have degree at least $d.$
If there exists a vertex $v \in J$ with degree less than $\frac{1}{\sqrt{\ell}}\frac{d}{n} |J|$ in $G[J]$ and degree more than $d$ in the current container, then selecting it and operation $\downarrow$ removes $\sqrt{\ell} \cdot \frac{n}{|J|}$ times more vertices from the container than from $J.$
Otherwise, there are only at most $\frac{\epsilon}{\sqrt{\ell}\rho^2}\frac{n}{d}|J|$ vertices in $J$ with degree larger than $d$ in the current container based on the sparsity of $G[J].$
Hence, selecting the vertex in $J$ with maximum degree in the current container that is less than $d,$ and operation $\uparrow,$ again removes roughly $\sqrt{\ell} \cdot \frac{n}{|J|}$ times more vertices from the container than from $J.$

In order to prove the stronger containment properties of \cref{lem:GCL}, we start by letting $r \cdot \frac{n}{|J|}$ be the smallest ratio achieved throughout the running of the procedure (by \cref{lem:weakContainment} we have that $r \geq \sqrt{\ell},$ but in many cases it may be much larger).
Generalizing the argument above, we can show that $C_{F}$ contains at least roughly $(1-\frac{1}{r})|J|$ of $J.$
Further, we can also show that the maximum size of the fingerprint is roughly $\frac{\sqrt{\ell}}{r}\frac{\rho^2}{\epsilon}$ (note: in order to show this we require that the greedy procedure described above only selects fingerprint vertices that remove some minimum threshold of vertices in the container).
This means that if $|C_{F}|$ is smaller than roughly $(1-\frac{\sqrt{\ell}}{r})\rho n$ then we have the desired bounds in the conclusions of the \cref{lem:GCL} and so no revision is necessary and we set $R=\bot.$

In the case that $|C_{F}|$ is larger than $(1-\frac{\sqrt{\ell}}{r})\rho n,$ we show (in \cref{lem:revision}) that the container $C_{F}$ can be revised to $C_{F}'$ by selecting a vertex $v \in J$ and removing all vertices from $C_{F}$ that had higher degree than $v$ at some time step $i$ of the procedure, and importantly, this revised container $C_{F}'$ satisfies the desired conclusions of the lemma.
We give more details below on how to select such a vertex, but before that we show how this implies \cref{lem:GCL}.
We make two important observations.
First, even though $C_F$ is constructed through a procedure that depends on $J,$ $C_{F}$ can be reconstructed solely from $F:$ initialize a container as $V$ and then iteratively run the operations from $F.$
Also, $C_{F}'$ can be constructed solely from $F$ and $R:$ construct $C_{F}$ and then, if $R =(i,v),$ revise $C_{F}$ by removing all vertices that had higher degree than $v$ at step $i.$
This is exactly what our container function (\cref{alg:Container}) does on input $(F,R).$
To summarize, for any sparse subgraph $G[J]$ we construct a fingerprint $(F,R)$ such that $C(F,R)$ satisfies the conclusions of the lemma, and the collection $\mathcal{F}$ is is the set of all fingerprints formed in this way.

\paragraph{Selecting a Revision Vertex.}
Consider the time step $i$ of the procedure that achieved the smallest ratio $r \cdot \frac{n}{|J|}$ of vertices removed from the container divided by vertices of $J$ removed from the container, and let the container at that time step be $D.$
To illustrate the key ideas, for now we will assume that $|D| \approx n.$
If $|C_{F}|$ is large then we can make an important observation about the structure of $D$ because $C_F \subseteq D$ and $G[C_F]$ is sparse.
In particular, we can show that either $|E(C_F,D \setminus C_F)|$ or $|E(D\setminus C_F)|$ is relatively large because $G$ is $\epsilon$-far from \rhoIndepSet.
In the case that $|E(D\setminus C_F)|$ is large, we find that, at time step $i,$ the fingerprint procedure should have been able to find a vertex achieving ratio larger than $r \cdot \frac{n}{|J|},$ which is a contradiction.
In the case that $|E(C_F,D\setminus C_F)|$ is large, we show that many vertices in $C_F$ have high degree in $G[D],$ and by removing all vertices in $C_F$ with high degree in $G[D]$ we remove a lot of $C_F$ without removing too much of $J.$

In slightly more detail, if $|C(F)| > (1-\frac{\sqrt{\ell}}{r})\rho n$ then we show that either $|E(C_F,D \setminus C_F)| \gtrsim \frac{\epsilon r}{\sqrt{\ell}\rho^2} |C_F| \cdot n$ or $|E(D\setminus C_F)| \gtrsim \frac{\epsilon r^2}{\ell \rho^2} n^2$ because otherwise we can add a random set of $\frac{\sqrt{\ell}}{r}\rho n$ vertices to $C_F$ without adding too many edges and find that $G$ is not $\epsilon$-far from \rhoIndepSet.

In order to use these lower bounds on the edges, we again use \cref{lem:degreeThresholdExistence} to argue about the degrees of vertices in dense subgraphs.
In particular, by \cref{lem:degreeThresholdExistence}, in the case that there are a lot of edges in $E(D\setminus C_F)$ there exists a degree threshold $d$ such that at least roughly $\frac{1}{d} \frac{\epsilon r^2}{\ell \rho^2} n^2$ vertices in $D \setminus C_F$ have degree larger than $d$ in $G[D].$
If there exists a vertex in $J$ with degree at least $d$ in $G[D]$ and degree smaller than roughly $\frac{1}{r} \frac{d}{n} |J|$ in $G[J]$ then the fingerprint procedure could have selected this vertex and the operation $\downarrow$ to remove $r \frac{n}{|J|}$ times more vertices from the container than from $J,$ which is a contradiction since the ratio at this time step was $r \cdot \frac{n}{|J|}.$
Otherwise, all the vertices in $J$ with degree at least $d$ in $G[D]$ have degree at least roughly $\frac{1}{r}\frac{d}{n} |J|$ in $G[J],$ and so there can only be at most $\frac{\epsilon r}{\ell \rho^2} \frac{n}{d} |J|$ of them based on the sparsity of $G[J].$
In this case, the fingerprint procedure could have selected the vertex in $J$ with highest degree in $D$ that is less than $d,$ and the operation $\uparrow$ to remove $r \frac{n}{|J|}$ times more vertices from the container than from $J,$ which again is a contradiction.

Since the first case leads to contradiction, it must be the case that $|E(C_F,D\setminus C_F)|$ is large.
In this case, we again use \cref{lem:degreeThresholdExistence} to find that there is a degree threshold $d$ such that at least roughly $\frac{1}{d}\frac{\epsilon r}{\sqrt{\ell}\rho^2} |C_F| n$ vertices in $C$ have degree larger than $d$ in $G[D].$
Using a similar argument, there can only be at most $\frac{\epsilon r}{\rho^2 \ell}\frac{n}{d} |J|$ vertices in $J$ with degree larger than $d$ in $G[D]$ otherwise the fingerprint procedure should have achieved a larger removal ratio, and so by removing vertices from $C_F$ with degree larger than $d$ in $G[D]$ we remove roughly $\sqrt{\ell} \cdot \frac{|C_F|}{|J|}$ times more vertices from $C_F$ than vertices from $J.$

\subsection{Fingerprint and Container Procedures}
\label{sect:newProcedures}
The greedy procedure used to construct the fingerprint sequence $F$ for a sparse set $J$ is in \cref{alg:Fingerprint}.
We initialize the container to $V,$ and the fingerprint to an empty sequence.
Then, in each step, we consider each $v \in J$ and the two possible operations of ``remove neighbours" or ``remove higher degree vertices".
For the first type of operation, the number of vertices that can be removed from the current container $C_{t-1}$ is $\deg_{G[C_{t-1}\cup \{v\}]}(v),$ where we have to use $C_{t-1}\cup \{v\}$ because $v$ may not be in $C_{t-1},$ and the number of vertices that would be removed from $J$ is at most $\deg_{G[J]}(v).$
Similarly, for the second type of operation, the vertices that can be removed from the current container $C_{t-1}$ are $\{u \in C_{t-1} : |N(u) \cap C_{t-1}| > |N(v) \cap C_{t-1}|\},$ which from now on we will denote by $(C_{t-1})_{\uparrow v},$ and the number of vertices in $J$ that would be removed in the process is at most $|(C_{t-1})_{\uparrow v} \cap J|.$
We desire the maximum ratio between the number of vertices removed in $C_{t-1}$ and the number of vertices removed in $J,$ however there is a small challenge.
There may be a vertex in $J$ that maximizes this ratio, but only removes a very small number of vertices in $J$ (say just $1$ vertex) and a very small number of vertices in $C_{t-1},$ and this would cause the procedure to take too long.
To handle this, we take the maximum of the vertices removed in $J$ with a threshold $\frac{\epsilon}{\sqrt{\ell}\rho^2}|J|.$
We select this threshold to give the desired trade offs in the fingerprint size and the amount of $J$ that is contained in the final container.
Formally, we define the maximal removal ratio as follows.

\begin{definition}
    \label{def:MRR}
    Given a graph $G=(V,E),$ subsets of vertices $J,C \subseteq V$ and parameters $\epsilon,\rho,\ell,$ define the \emph{maximum removal ratio}, denoted ${\rm MRR}_{G}(J,C),$ as \[ {\rm MRR}_{G}(J,C) \coloneqq \max_{v \in J} \max\left(\frac{\deg_{G[C \cup \{v\}]}(v)}{\max(\deg_{G[J]}(v),\frac{\epsilon}{\sqrt{\ell} \rho^2}|J|)},\frac{|C_{\uparrow v}|}{\max(|C_{\uparrow v} \cap J|,\frac{\epsilon}{\sqrt{\ell} \rho^2}|J|)}\right).\]
    While this quantity depends on $\epsilon,\rho$ and $\ell,$ we suppress this dependence in the notation ${\rm MRR}_G(J,C)$ for clarity.
\end{definition}

Overall, in each step of the fingerprint procedure we select the vertex $v \in J$ and operation that achieves the maximum removal ratio, add it to the fingerprint and run the operation on the container, and repeat.

\begin{algorithm}
    \caption{\sc{Fingerprint Generator}}
    \label{alg:Fingerprint}
    \smallskip
    \KwIn{A graph $G=(V,E),$ and $J \subseteq V$ for which $G[J]$ has fewer than $\frac{\epsilon}{\ell \rho^2}|J|^2$ edges}
    \medskip
    Initialize $t=0, F \gets ()$ and $C_0 \gets V$\;

    \While{$G[C_t]$ has more than $\epsilon n^2/4$ edges} {
        $t \gets t+1$\;
        $v_t \gets$ $\argmax_{v \in J} \max\left(\frac{\deg_{G[C_{t-1} \cup \{v\}]}(v)}{\max(\deg_{G[J]}(v),\frac{\epsilon}{\sqrt{\ell} \rho^2}|J|)},\frac{|(C_{t-1})_{\uparrow v}|}{\max(|(C_{t-1})_{\uparrow v} \cap J|,\frac{\epsilon}{\sqrt{\ell} \rho^2}|J|)}\right)$\;
        \If{$v_t$ comes from the first type}{
            $C_t \gets C_{t-1} \setminus N(v_t)$\;
            Append $(v_t,\downarrow)$ to $F$\;
        }
        \Else{
            $C_t \gets C_{t-1} \setminus (C_{t-1})_{\uparrow v_t}$\;
            Append $(v_t,\uparrow)$ to $F$\;
        }
    }
    Return $F$\;
\end{algorithm}

\begin{algorithm}
    \caption{\sc{Container Generator}}
    \label{alg:Container}
    \smallskip
    \KwIn{A graph $G=(V,E),$ and fingerprint sequence $F$ and revision $R$}
    \medskip
    Initialize $C_0 \gets V$\;

    \For{$t=1,2,\dots,|F|$} {
        \If{$F[t]=(v,\downarrow),$ for some $v \in V$}{$C_t \gets C_{t-1} \setminus N(v)$}
        \Else{$C_t \gets C_{t-1} \setminus (C_{t-1})_{\uparrow v}$}
    }

    \If{$R=(i,u)$} {
        Return $C_{|F|} \setminus (C_i)_{\uparrow u}$
    }
    \Else{
        Return $C_{|F|}$\;
    }
\end{algorithm}

The first part of the container function, \cref{alg:Container}, essentially simulates the container produced in the fingerprint procedure.
In particular, given a fingerprint sequence $F$ it repeatedly runs the operations on the container described by the fingerprint, and eventually finds container $C_{|F|}.$
Afterwards, if a revision pair $(i,u)$ is given, then returns $C_{|F|} \setminus (C_i)_{\uparrow u},$ and otherwise returns $C_{|F|}.$

\subsection{An Initial Containment Lemma}
\label{sect:initialContainment}
In this section we prove that \cref{alg:Fingerprint} makes reasonable progress in every step, which also implies an initial weak containment property on the final container constructed in \cref{alg:Fingerprint}.
This will be used as a starting point in the proof of the graph container lemma (\cref{lem:GCL}) in \cref{sect:proofGCL}.

Before we state and prove the lemma, we require the following \cref{lem:degreeThresholdExistence} which says that if there are subsets of vertices $U,W$ in $V$ such that $|E(U,W)|$ is large, then there are a relatively large number of vertices in $U$ that have a relatively large number of neighbours in $W.$
In particular, there is a threshold $d \leq |W|$ such that at least roughly $|E(U,W)|/d$ vertices in $U$ have more than $d$ neighbours in $W.$
\cref{lem:degreeThresholdExistence} will be used in the proof of \cref{lem:weakContainment} as well as \cref{lem:revision} in the next section.

\begin{lemma}
    \label{lem:degreeThresholdExistence}
    Let $G=(V,E)$ be a graph and let $U,W \subseteq V.$
    If $|E(U,W)| \geq m_{UW}$ then there exists $\frac{m_{UW}}{2|U|} \leq d \leq |W|$ such that the number of vertices in $U$ with more than $d$ neighbours in $W$ is at least \[\frac{m_{UW}}{4d \log(2|U||W|/m_{UW})}.\]
\end{lemma}
\begin{proof}
    For ease of presentation let $\delta=\frac{m_{UW}}{|U||W|},$ which is roughly the edge density between $U$ and $W.$
    Let $d_k=\frac{1}{2^k}|W|$ for $k=0,\dots,\log(2/\delta).$
    Observe that for all $k,$ $\frac{m_{UW}}{2|U|} \leq d_k \leq |W|.$
    We claim that there exists a $k$ such that the number of vertices in $U$ with more than $d_k$ neighbours in $W$ is at least $\frac{m_{UW}}{4\log(2/\delta) d_k}.$
    To prove the claim, we suppose this is not the case, and find a contradiction.
    For $k=1,\dots,\log(2/\delta),$ let $U_k$ be the vertices in $U$ with more than $d_k$ neighbours and most $d_{k-1}$ neighbours in $W,$ and let $U_{\log(2/\delta)+1}$ be the vertices with at most $d_{\log(2/\delta)}=\frac{m_{UW}}{2|U|}$ neighbours in $W.$
    Then $U_1,\dots,U_{\log(2/\delta)+1}$ partition $U,$ and so the number of edges in $E(U,W)$ is at most \[|U_{\log(2/\delta)+1}| \cdot \frac{m_{UW}}{2|U|} + \sum_{k=1}^{\log(2/\delta)} |U_k| \cdot d_{k-1} < \frac{m_{UW}}{2} + \sum_{k=1}^{\log(2/\delta)} \frac{m_{UW}}{4\log(2/\delta) d_k} d_{k-1} \leq m_{UW},\] where the first inequality uses the trivial bound of $|U_{\log(2/\delta)+1}| \leq |U|,$ as well as the bound $|U_k| < \frac{m_{UW}}{4 d_k \log(2/\delta)}$ for all $k,$ which holds because of the assumption that there are fewer than $\frac{m_{UW}}{4\log(2/\delta) d_k}$ vertices in $U$ that have more than $d_k$ neighbours in $W.$
    Hence, we have that $|E(U,W)| < m_{UW},$ which is a contradiction, and so there exists $\frac{m_{UW}}{2|U|} \leq d \leq |W|$ such that the number of vertices in $U$ with more than $d$ neighbours in $W$ is at least $\frac{m_{UW}}{4d \log(2|U||W|/m_{UW})}.$
\end{proof}

We can now state and prove an initial containment lemma.
We show that if $G$ is $\epsilon$-far from \rhoIndepSet, then every step of \cref{alg:Fingerprint} makes progress, and we use this to give an initial upper bound on the size of the fingerprint sequence and an initial lower bound on the fraction of $J$ that is in the final container of \cref{alg:Fingerprint}.

\begin{lemma}
    \label{lem:weakContainment}
    Let $\epsilon \leq \rho^2/2$ and $\ell \geq 256.$ Let $G=(V,E)$ be $\epsilon$-far from having a $\rho n$ independent set, let $J \subseteq V$ be a set of vertices for which $G[J]$ has fewer than $\frac{\epsilon}{\ell \rho^2} |J|^2$ edges, and let $C$ be the final container in \cref{alg:Fingerprint}.
    Then for each step $t$ of the procedure, 
    \begin{equation}
        \label{eqn:weakGCLProgress1}
        {\rm MRR}_G(J,C_{t-1}) \geq \frac{\sqrt{\ell} \max(|C_{t-1}|,\rho n)}{16 \log(8\rho/\epsilon)|J|}
    \end{equation}
    and $|C \cap J| \geq \left(1-\frac{16 \log^2(8\rho/\epsilon)}{\sqrt{\ell}}\right)|J|$ and $|F| \leq \frac{16 \rho^2 \log^2(8\rho/\epsilon)}{\epsilon}.$
\end{lemma}
\begin{proof}
Let $t$ be a time step during the procedure of \cref{alg:Fingerprint}.
We will first prove \cref{eqn:weakGCLProgress1}.
First observe that $G[C_{t-1}]$ has at least $\frac{\epsilon}{4\rho^2}\left(\max(|C_{t-1}|,\rho n)\right)^2$ edges.
When $|C_{t-1}| < \rho n$ this follows because $C_{t-1}$ has at least $\epsilon n^2/4$ edges, and when $|C_{t-1}| \geq \rho n$ this follows because $G$ is $\epsilon$-far from having a $\rho n$ independent set, and so any set of size larger than $\rho n$ has edge density at least $\frac{\epsilon}{\rho^2}.$
Hence, by applying \cref{lem:degreeThresholdExistence} with $U=W=C_{t-1},$ there exists $d$ such that $\frac{\epsilon}{8\rho^2 |C_{t-1}|}\left(\max(|C_{t-1}|,\rho n)\right)^2 \leq d \leq |C_{t-1}|$ and at least $\frac{\epsilon\left(\max(|C_{t-1}|,\rho n)\right)^2}{16 \rho^2 d \log(8\rho^2/\epsilon)}$ vertices in $C_{t-1}$ have degree larger than $d$ in $G[C_{t-1}].$

Consider the vertices in $J$ with more than $d$ neighbours in $C_{t-1},$ we consider three cases.
First, if there exists such a vertex $v \in J$ with degree smaller than $\frac{\epsilon}{\sqrt{\ell}\rho^2}|J|$ in $G[J],$ then 
\[
    \frac{\deg_{G[C_{t-1} \cup \{v\}]}(v)}{\max(\deg_{G[J]}(v),\frac{\epsilon}{\sqrt{\ell} \rho^2}|J|)} \geq \frac{d}{\frac{\epsilon}{\sqrt{\ell} \rho^2}|J|} \geq \frac{\frac{\epsilon}{8\rho^2 |C_{t-1}|}\left(\max(|C_{t-1}|,\rho n)\right)^2}{\frac{\epsilon}{\sqrt{\ell} \rho^2}|J|} \geq \frac{\sqrt{\ell} \max(|C_{t-1}|,\rho n)}{8 |J|},
    \]
and so \cref{eqn:weakGCLProgress1} holds.
Next, if there exists such a vertex $v \in J$ with degree in $G[J]$ at least $\frac{\epsilon}{\sqrt{\ell}\rho^2}|J|$ but at most $\frac{2d |J|}{\sqrt{\ell} \max(|C_{t-1}|,\rho n)},$ then again observe that 
\[
    \frac{\deg_{G[C_{t-1} \cup \{v\}]}(v)}{\max(\deg_{G[J]}(v),\frac{\epsilon}{\sqrt{\ell} \rho^2}|J|)} \geq \frac{d}{\frac{2d|J|}{\sqrt{\ell}\max(|C_{t-1}|,\rho n)}} \geq \frac{\sqrt{\ell} \max(|C_{t-1}|,\rho n)}{2 |J|},
\]
and so we again we find that \cref{eqn:weakGCLProgress1} holds.
Finally, in the case that neither of the first two cases hold, we count the number of vertices in $J$ with more than $d$ neighbours in $C_{t-1}.$
Since neither of the first two cases hold, then any such vertex has degree more than $\frac{2d |J|}{\sqrt{\ell} \max(|C_{t-1}|,\rho n)}$ in $G[J],$ and so there can be at most $\frac{\epsilon \max(|C_{t-1}|,\rho n) |J|}{\sqrt{\ell}\rho^2 d}$ such vertices because of the sparsity of $G[J].$
Let $v$ be the vertex in $J$ that has the most number of neighbours in $C_{t-1}$ that is at most $d.$
Such a vertex exists because $\frac{\epsilon \max(|C_{t-1}|,\rho n) |J|}{\sqrt{\ell}\rho^2 d} \leq |J|/2$ using the lower bound on $d$ and that $\ell \geq 256.$
Then
\[
    \frac{|(C_{t-1})_{\uparrow v}|}{\max(|(C_{t-1})_{\uparrow v} \cap J|,\frac{\epsilon}{\sqrt{\ell} \rho^2}|J|)} \geq \frac{\frac{\epsilon\left(\max(|C_{t-1}|,\rho n)\right)^2}{16 \rho^2 d \log(8\rho^2/\epsilon)}}{\max(\frac{\epsilon \max(|C_{t-1}|,\rho n) |J|}{\sqrt{\ell}\rho^2 d},\frac{\epsilon}{\sqrt{\ell} \rho^2}|J|)} \geq \frac{\sqrt{\ell} \max(|C_{t-1}|,\rho n)}{16 \log(8\rho^2/\epsilon)|J|},
\]
where the last inequality uses that $d \leq |C_{t-1}|$ in the case that $\frac{\epsilon \max(|C_{t-1}|,\rho n) |J|}{\sqrt{\ell}\rho^2 d} \leq \frac{\epsilon}{\sqrt{\ell}\rho^2}|J|.$
Using that $\log(8\rho^2/\epsilon) \leq \log(8\rho/\epsilon)$ completes the proof of \cref{eqn:weakGCLProgress1}.

Now, for any step $t$ of the procedure, let $\beta_t |J|$ be the number of vertices removed from $J$ when constructing $C_t.$
\cref{eqn:weakGCLProgress1} can be written as the following lower bound on the number of vertices removed from $C_{t-1}$ to construct $C_t:$
\begin{equation}
    \label{eqn:weakGCLProgress}
    |C_{t-1}| - |C_t| \geq \frac{\sqrt{\ell}}{16 \log(8\rho/\epsilon)} \frac{\max\left(\beta_t|J|,\frac{\epsilon}{\sqrt{\ell}\rho^2}|J|\right)}{|J|} \max(|C_{t-1}|,\rho n).
\end{equation}

Let $t^*$ be the smallest index for which $|C_{t^*}| \leq \rho n.$
Break up the procedure into three parts: when $t < t^*,$ when $t=t^*$ and $t > t^*.$
First, for each $t< t^*,$ rearrange \cref{eqn:weakGCLProgress} to 
\[
    |C_t| \leq \left(1 - \frac{\sqrt{\ell}}{16 \log(8\rho/\epsilon)} \max\left(\beta_t, \frac{\epsilon}{\sqrt{\ell}\rho^2}\right)\right)|C_{t-1}|,
\]
and by applying for every $t < t^*$ we find that 
\[
    |C_{t^*-1}|\leq \prod_{t=1}^{t^*-1} \left(1-\frac{\sqrt{\ell}}{16 \log(8\rho/\epsilon)} \max\left(\beta_t, \frac{\epsilon}{\sqrt{\ell}\rho^2}\right)\right) \cdot |C_0|.
\]
Since $|C_{t^*-1}|\geq \rho n$ and $|C_0|=n,$ and using the upper bound $(1-x) \leq e^{-x},$ we find
\begin{equation}
    \label{eqn:weakGCLFinalBound}
    \rho \leq \exp\left(-\frac{\sqrt{\ell}}{16 \log(8\rho/\epsilon)} \sum_{t=1}^{t^*-1} \max\left(\beta_t, \frac{\epsilon}{\sqrt{\ell}\rho^2}\right)\right).
\end{equation}
Using that $\max(\cdot)$ is at least as large as the first operand, \cref{eqn:weakGCLFinalBound} implies that $\sum_{t=1}^{t^*-1} \beta_t \leq \frac{16 \log(8\rho/\epsilon)\ln(1/\rho)}{\sqrt{\ell}}.$
Similarly, using that $\max(\cdot)$ is at least as large as the second operand, \cref{eqn:weakGCLFinalBound} implies that $t^* \leq \frac{16 \rho^2\log(8\rho/\epsilon) \ln(1/\rho)}{\epsilon}+1.$

In the case $t=t^*,$ \cref{eqn:weakGCLProgress} implies that $\beta_t \leq \frac{16 \log(8\rho/\epsilon)}{\sqrt{\ell}}$ because $|C_{t^*}| \geq 0.$
Finally, for each $t > t^*,$ \cref{eqn:weakGCLProgress} states that 
\[
    |C_{t-1}| - |C_t| \geq \frac{\sqrt{\ell}}{16 \log(8\rho/\epsilon)} \frac{\max\left(\beta_t|J|,\frac{\epsilon}{\sqrt{\ell}\rho^2}|J|\right)}{|J|} \rho n.
\]
Summing up over all $t> t^*$ we find
\[
    |C_{t^*}| - |C_{|F|}| \geq \sum_{t=t^*+1}^{|F|}\frac{\sqrt{\ell}}{16 \log(8\rho/\epsilon)} \frac{\max\left(\beta_t|J|,\frac{\epsilon}{\sqrt{\ell}\rho^2}|J|\right)}{|J|} \rho n.
\]
Since $|C_{t^*}| \leq \rho n$ and $|C_{|F|}| \geq 0,$ lower bounding the max by the first operand gives $\sum_{t=t^*+1}^{|F|} \beta_t \leq \frac{16 \log(8\rho/\epsilon)}{\sqrt{\ell}},$ and lower bounding the max by the second operand gives $|F|-t^* \leq \frac{16 \log(8\rho/\epsilon) \rho^2}{\epsilon}.$
Combining, we have that $\sum_{t=1}^{|F|} \beta_t \leq \frac{16 \log(8\rho/\epsilon) (\ln(1/\rho)+2)}{\sqrt{\ell}}$ and $|F| \leq \frac{16 \rho^2 \log(8\rho/\epsilon) (\ln(1/\rho)+2)}{\epsilon},$ which can be upper bounded as $\sum_{t=1}^{|F|} \beta_t \leq \frac{16 \log^2(8\rho/\epsilon)}{\sqrt{\ell}}$ and $|F| \leq \frac{16 \rho^2\log^2(8\rho/\epsilon)}{\epsilon}$ by using $\epsilon \leq \rho^2/2,$ which completes the proof of the lemma.
\end{proof}

\subsection{Revision Lemma}
\label{sect:revisionLemma}
Our next lemma shows the existence of a ``revision" in the case that the final container in \cref{alg:Fingerprint} is too large relative to the fraction of $J$ that is contained in it.
In particular, suppose $C$ is the final container in \cref{alg:Fingerprint} and $D$ is the container at some step of the procedure in \cref{alg:Fingerprint}.
This lemma roughly says that if $G[C]$ is large and sparse and there wasn't a good option for the fingerprint selection in \cref{alg:Fingerprint}, then there must be a vertex $v$ in $J$ such that there are a relatively high number of vertices in $C$ with higher degree than $v$ in $D.$
This vertex can be used to reduce the size of the container $C.$

\begin{lemma}
    \label{lem:revision}
    Let $\epsilon \leq \rho^2/2$ and let $\ell \geq 16.$
    Let $G=(V,E)$ be $\epsilon$-far from having a $\rho n$ independent set, let $J \subseteq V$ be a set of vertices for which $G[J]$ has fewer than $\frac{\epsilon}{\ell \rho^2} |J|^2$ edges, and let $C,D \subseteq V$ be subsets of vertices such that $C \subseteq D,$ $|E(C)| \leq \epsilon n^2/4,$ and $|D| \geq \rho n.$
    If for some $\rr > \frac{\sqrt{\ell}}{16\log(8\rho/\epsilon)}:$ $|C|\geq \left(1-\frac{\sqrt{\ell}}{16 \rr \log(8\rho/\epsilon)}\right)\rho n,$ and
    \begin{equation}
        \label{eqn:ratioBound}
        {\rm MRR}_G(J,D) \leq \frac{\rr \left(|D|-\left(1-\frac{\sqrt{\ell}}{16 \rr \log(8\rho/\epsilon)}\right)\rho n\right)}{|J|}
    \end{equation}
    then there exists $v \in J$ and $\gamma \geq \frac{\sqrt{\ell}}{16 \rr \log(8\rho/\epsilon)}$ such that $|C \setminus D_{\uparrow v}| \leq (1-\gamma)\rho n$ and $|D_{\uparrow v} \cap J| \leq \frac{16 \gamma \log^2(8\rho/\epsilon)}{\sqrt{\ell}} |J|.$
\end{lemma}
\begin{proof}
    Let $Y$ be any subset of $C$ of size $\left(1-\frac{\sqrt{\ell}}{16 \rr \log(8\rho/\epsilon)}\right)\rho n.$
    We first claim that for any $d \geq \frac{\epsilon \rr |D \setminus Y|}{\sqrt{\ell}\rho^2},$ the number of vertices in $J$ with more than $d$ neighbours in $D$ is at most $\frac{2\epsilon \rr |J||D \setminus Y|}{d \ell \rho^2}.$
    To show this, observe that any such vertex $v \in J$ has degree larger than $\frac{\epsilon}{\sqrt{\ell}\rho^2}|J|$ in $G[J]$ because otherwise $\frac{\deg_{G[D \cup \{v\}]}(v)}{\max(\deg_{G[J]}(v),\frac{\epsilon}{\sqrt{\ell} \rho^2}|J|)} > \frac{d}{\frac{\epsilon}{\sqrt{\ell}\rho^2}|J|} \geq \frac{\rr |D \setminus Y|}{|J|},$ which contradicts \cref{eqn:ratioBound}.
    Further, any such vertex $v \in J$ has degree larger than $\frac{d|J|}{\rr |D \setminus Y|}$ in $G[J]$ because otherwise we find that $\frac{\deg_{G[D \cup \{v\}]}(v)}{\max(\deg_{G[J]}(v),\frac{\epsilon}{\sqrt{\ell} \rho^2}|J|)} > \frac{\rr |D \setminus Y|}{|J|},$ again contradicting \cref{eqn:ratioBound}.
    Hence, since any such vertex $v \in J$ has degree larger than $\frac{d|J|}{\rr |D \setminus Y|}$ in $G[J],$ then there can be at most $\frac{2\epsilon \rr |J| |D \setminus Y|}{d \ell \rho^2}$ of them, based on the sparsity of $J,$ which completes the proof of the claim.
    We also note that $\frac{2\epsilon \rr |J| |D \setminus Y|}{d \ell \rho^2} \leq |J|/2$ using the lower bound on $d$ and that $\ell \geq 16,$ and so there always exists $v \in J$ with at most $d$ neighbours in $D.$
    
    We now claim that $|E(Y,D\setminus Y)| > \frac{4 \log(8\rho/\epsilon) \epsilon \rr}{\sqrt{\ell}\rho^2} \rho n |D \setminus Y|.$
    Before proving this, we show how this implies the lemma.
    Apply \cref{lem:degreeThresholdExistence} (with $U=Y$ and $W=D \setminus Y$) to find there exists $d_Y \in \left[\frac{2\log(8\rho/\epsilon) \epsilon \rr}{\sqrt{\ell}\rho^2 |Y|} \rho n \cdot |D \setminus Y|, |D\setminus Y|\right]$ such that at least $\frac{\epsilon \rr}{\sqrt{\ell}\rho^2 d_Y} \rho n \cdot |D \setminus Y|$ vertices in $Y$ have degree larger than $d_Y$ in $G[D].$
    Let $v_Y$ be the vertex in $J$ that has largest degree in $G[D \cup \{v_Y\}]$ less than $d_Y.$
    Since $d_Y \geq \frac{\epsilon r |D \setminus Y|}{\sqrt{\ell}\rho^2}$, the first claim gives $|D_{\uparrow v_Y} \cap J| \leq \frac{2 \epsilon \rr |J||D \setminus Y|}{d_Y \ell \rho^2},$ and we find that
    \[
        \frac{|Y \cap D_{\uparrow v_Y}|}{\max(|D_{\uparrow v_Y} \cap J|,\frac{\epsilon}{\sqrt{\ell}\rho^2}|J|)} \geq \frac{\frac{\epsilon \rr}{\sqrt{\ell}\rho^2 d_Y} \rho n \cdot |D \setminus Y|}{\max\left(\frac{2 \epsilon \rr |J||D \setminus Y|}{d_Y \ell \rho^2},\frac{\epsilon}{\sqrt{\ell}\rho^2}|J|\right)} \geq \frac{\sqrt{\ell}\rho n}{16 \log(8\rho/\epsilon)|J|},
    \]
    where we use that $\rr \geq \frac{\sqrt{\ell}}{16 \log(8\rho/\epsilon)}$ and $|D\setminus Y| \geq d_Y.$
    
    Since this applies for any $Y \subseteq C$ of size $\left(1-\frac{\sqrt{\ell}}{16 \rr \log(8\rho/\epsilon)}\right)\rho n,$ consider all such choices of $Y$ and let $v$ be the $v_Y$ with the minimum number of neighbours in $D.$
    We claim that $v$ satisfies the conclusion of the lemma.
    First, if $|J \cap D_{\uparrow v}| \geq \frac{\log(8\rho/\epsilon)}{\rr}|J|,$ then selecting $\gamma$ so that $|J \cap D_{\uparrow v}| = \frac{16 \gamma \log(8\rho/\epsilon)}{\sqrt{\ell}} |J|$ suffices because at least $\gamma \rho n \geq \frac{\sqrt{\ell}}{16 \rr}\rho n$ vertices are removed from the associated $Y,$ which is a subset of $C,$ and $|C| <\rho n.$
    Otherwise, use $\gamma=\frac{\sqrt{\ell}}{16 \rr \log(8\rho/\epsilon)}$ and observe that at least $1$ vertex is removed from every set $Y$ since $v$ was selected to be the lowest degree such vertex.
    Hence, $|C \setminus D_{\uparrow v}| < \left(1-\frac{\sqrt{\ell}}{16 \rr \log(8\rho/\epsilon)}\right)\rho n=(1-\gamma)\rho n.$
    
    Now we prove the desired bound on $|E(Y,D\setminus Y)|$ for any $Y \subseteq C$ with $|Y|=\left(1-\frac{\sqrt{\ell}}{16 \rr \log(8\rho/\epsilon)}\right)\rho n.$
    To start, we claim that $|E(D \setminus Y)| \leq \frac{64 \log^2(8\rho/\epsilon) \epsilon \rr^2}{\ell\rho^2} |D \setminus Y|^2.$
    To show this, suppose not and find a contradiction.
    Using \cref{lem:degreeThresholdExistence} (with $U=W=D \setminus Y$), there exists a $d$ satisfying $\frac{32 \log^2(8\rho/\epsilon) \epsilon \rr^2}{\ell \rho^2}{ |D\setminus Y| }\leq d \leq |D \setminus Y|$ such that at least $\frac{16 \log(8\rho/\epsilon) \epsilon \rr^2}{\ell \rho^2 d} |D \setminus Y|^2$ vertices in $D \setminus Y$ have degree larger than $d$ in $G[D \setminus Y]$ (which also means they have degree larger than $d$ in $G[D]$).
    Again using the first claim (since $d$ is sufficiently large), we find that at most $\frac{2\epsilon \rr |J| |D \setminus Y|}{d \ell \rho^2}$ vertices in $J$ have more than $d$ neighbours in $D.$
    
    Let $v$ be the vertex in $J$ having the highest number of neighbours in $D$ that is less than $d.$
    Then $|D_{\uparrow v} \cap J| \leq \frac{2\epsilon \rr |J| |D \setminus Y|}{d \ell \rho^2},$ and we find that 
    \[
        \frac{|(D \setminus Y) \cap D_{\uparrow v}|}{\max(|D_{\uparrow v} \cap J|,\frac{\epsilon}{\sqrt{\ell}\rho^2}|J|)} \geq \frac{\frac{16 \log(8\rho/\epsilon) \epsilon \rr^2}{\ell \rho^2 d} |D \setminus Y|^2}{\max(\frac{2 \epsilon \rr |J||D \setminus Y|}{d \ell \rho^2},\frac{\epsilon}{\sqrt{\ell}\rho^2}|J|)} > \frac{\rr |D \setminus Y|}{|J|},
    \]
    where we have to use that $\rr > \frac{\sqrt{\ell}}{16\log(8\rho/\epsilon)}$ and $|D\setminus Y| \geq d,$ which is a contradiction with \cref{eqn:ratioBound}.
    Hence $|E(D \setminus Y)| \leq \frac{64 \log^2(8\rho/\epsilon) \epsilon \rr^2}{\ell\rho^2} |D \setminus Y|^2.$

    Now, to show the lower bound on $|E(Y,D\setminus Y)|,$ let $R$ be a random set of size $\frac{\sqrt{\ell}}{16 \log(8\rho/\epsilon) \rr}\rho n$ vertices from $D \setminus Y.$
    We will count the number of edges in $G[R \cup Y]$ in expectation, and use this to get the desired bound on $|E(Y,D\setminus Y)|.$
    Using the upper bound on $|E(D \setminus Y)|,$ the expected number of edges in $G[R]$ is at most $\frac{64 \log^2(8\rho/\epsilon) \epsilon \rr^2}{\ell\rho^2} |R|^2 \leq \epsilon n^2/4.$
    The expected number of edges between $R$ and $Y$ is $|E(Y,D \setminus Y)| \cdot \frac{|R|}{|D \setminus Y|} = |E(Y,D\setminus Y)| \cdot \frac{\sqrt{\ell} \rho n}{16 \log(8\rho/\epsilon) \rr |D \setminus Y|}.$
    Finally, using that the number of edges in $G[Y]$ is at most $\epsilon n^2/4,$ then the expected number of edges in $G[R \cup Y]$ is at most $\epsilon n^2/4 + \epsilon n^2/4 + |E(Y,D\setminus Y)| \cdot \frac{\sqrt{\ell} \rho n}{16 \log(8\rho/\epsilon) \rr |D \setminus Y|}.$
    Since $R$ was selected randomly there exists a set $R$ for which there are at most that many edges in $G[R \cup Y],$ and since $|Y \cup R|=\rho n$ and $G$ is $\epsilon$-far from having a $\rho n$ independent set, then $|E(Y,D\setminus Y)| > \frac{4 \log(8\rho/\epsilon) \rr \epsilon}{\sqrt{\ell}\rho^2} \rho n |D \setminus Y|.$
\end{proof}

\subsection{Proof of Graph Container Lemma}
\label{sect:proofGCL}
In this section we complete the proof of \cref{lem:GCL}.

{\renewcommand\footnote[1]{}\lemGCL*} 

\begin{proof}
Let $J$ be any subset of $V$ for which $G[J]$ has fewer than $\frac{\epsilon}{\ell \rho^2}|J|^2$ edges.
We will construct a fingerprint $(F,R)$ that, along with container function \cref{alg:Container}, satisfies the conclusions of the lemma for set $J.$

Let $F$ be the fingerprint sequence constructed by \cref{alg:Fingerprint}, and let $C_{|F|}$ be the final container of \cref{alg:Fingerprint} run on $J.$
First observe that since $G[C_{|F|}]$ has at most $\epsilon n^2/4$ edges then $|C_{|F|}|< \left(1-\frac{\epsilon}{2\rho^2}\right)\rho n$ because otherwise by adding arbitrary vertices to $C_{|F|}$ until $|C_{F}|=\rho n$ (so at most $\frac{\epsilon}{2\rho^2}\rho n$ vertices added) we would find an induced subgraph on at least $\rho n$ vertices that has fewer than $\epsilon n^2$ edges, a contradiction.
Hence, let $|C_{|F|}|=(1-\alpha)\rho n,$ where $\frac{\epsilon}{2\rho^2} \leq \alpha \leq 1.$
By \cref{lem:weakContainment}, $|C_{|F|} \cap J| \geq \left(1-\frac{16\log^2(8\rho/\epsilon)}{\sqrt{\ell}}\right)|J|$ and $|F| \leq \frac{16 \rho^2 \log^2(8\rho/\epsilon)}{\epsilon}.$
Hence, if $\alpha \geq \frac12$ then set $R = \bot$ and observe that the container function \cref{alg:Container} outputs $C_{|F|}$ on input $(F,R)$ and satisfies the conclusions of the lemma for set $J.$
For the remainder of the proof we consider the case that $\alpha < \frac12.$

Let $t^*$ be the smallest index for which $|C_{t^*}| \leq \rho n.$
Break up the procedure into two parts: when $t \leq t^*$ and when $t > t^*.$
First, for $t\leq t^*$ we are going to argue about the step of the procedure that reaches the worst (smallest) ratio. 
Let $\rr$ be the smallest value such that there exists a step $t\leq t^*$ with 
\begin{equation}
    \label{eqn:worstRatio}
    {\rm MRR}_G(J,C_{t-1}) = \frac{\rr (|C_{t-1}|-(1-2\alpha)\rho n)}{|J|}.
\end{equation}

We claim that $t^* \leq \frac{\sqrt{\ell} \rho^2}{\epsilon}\frac{\log(2\rho/\epsilon)}{\rr}$ and $|C_{t^*} \cap J| \geq \left(1-\frac{\log(2\rho/\epsilon)}{\rr}\right)|J|.$
To show this, for each step $t\leq t^*,$ let $\beta_t |J|$ be the number of vertices of $J$ that are removed when constructing $C_{t}$ from $C_{t-1}.$
Then, based on the minimality of $\rr,$ we have that \[|C_{t-1}| - |C_t| \geq \max\left(\beta_t |J|, \frac{\epsilon}{\sqrt{\ell}\rho^2}|J|\right) \frac{\rr(|C_{t-1}|-(1-2\alpha)\rho n)}{|J|},\]
and by rearranging and subtracting $(1-2\alpha)\rho n$ from both sides, this can be rewritten as
\[
    |C_t|-(1-2\alpha)\rho n \leq \left(1 - \rr \max\left(\beta_t, \frac{\epsilon}{\sqrt{\ell}\rho^2}\right)\right) \left(|C_{t-1}|-(1-2\alpha)\rho n\right).
\]
Since this applies for all $t=1,\dots,t^*,$ then 
\[
    |C_{t^*}| -(1-2\alpha)\rho n \leq \prod_{t=1}^{t^*} \left(1-\rr \max\left(\beta_t, \frac{\epsilon}{\sqrt{\ell}\rho^2}\right)\right) \cdot (|C_0| - (1-2\alpha)\rho n).
\]
Since $|C_{t^*}| \geq (1-\alpha)\rho n$ then $|C_{t^*}| -(1-2\alpha)\rho n \geq \alpha \rho n \geq \frac{\epsilon}{2\rho}n.$
Also, since $\alpha < 1/2$ then $|C_0| - (1-2\alpha)\rho n < n.$
Substituting in, and using that $(1-x)\leq e^{-x},$ we find that
\begin{equation}
    \label{eqn:DkUpperBound}
    \frac{\epsilon}{2\rho} \leq \exp\left(-\rr \sum_{t=1}^{t^*} \max\left(\beta_t, \frac{\epsilon}{\sqrt{\ell}\rho^2}\right)\right).
\end{equation}
Using that the $\max(\cdot)$ in \cref{eqn:DkUpperBound} is at least as large as the first operand, we have that $\sum_{t=1}^{t^*} \beta_t \leq \frac{\ln(2\rho/\epsilon)}{\rr} \leq \frac{\log(2\rho/\epsilon)}{\rr}.$
Similarly, using that the $\max(\cdot)$ in \cref{eqn:DkUpperBound} is at least as large as the second operand, then $t^* \leq \frac{\sqrt{\ell} \rho^2 \ln(2\rho/\epsilon)}{\rr \epsilon} \leq \frac{\sqrt{\ell} \rho^2 \log(2\rho/\epsilon)}{\rr \epsilon}.$

For the remaining steps $t > t^*,$ by \cref{lem:weakContainment} we have that
\[
    |C_{t-1}| - |C_t| \geq \frac{\sqrt{\ell}}{16 \log(8\rho/\epsilon)} \frac{\max\left(\beta_t|J|,\frac{\epsilon}{\sqrt{\ell}\rho^2}|J|\right)}{|J|} \rho n.
\]
Summing up over all $t> t^*$ we find
\[
    |C_{t^*}| - |C_{|F|}| \geq \sum_{t=t^*+1}^{|F|}\frac{\sqrt{\ell}}{16 \log(8\rho/\epsilon)} \frac{\max\left(\beta_t|J|,\frac{\epsilon}{\sqrt{\ell}\rho^2}|J|\right)}{|J|} \rho n.
\]
Since $|C_{t^*}| \leq \rho n$ and $|C_{|F|}| \geq (1-\alpha)\rho n,$ lower bounding the $\max$ by the first operand gives $\sum_{t=t^*+1}^{|F|} \beta_t \leq \frac{16 \alpha \log(8\rho/\epsilon)}{\sqrt{\ell}},$ and lower bounding the max by the second operand gives $|F|-t^* \leq \frac{16 \alpha \log(8\rho/\epsilon) \rho^2}{\epsilon}.$
Combining, we have that $\sum_{t=1}^{|F|} \beta_t \leq \frac{\log(2\rho/\epsilon)}{r}+\frac{16 \alpha \log(8\rho/\epsilon)}{\sqrt{\ell}}$ and $|F| \leq \frac{\sqrt{\ell} \rho^2 \log(2\rho/\epsilon)}{\rr \epsilon}+ \frac{16 \alpha \rho^2\log(8\rho/\epsilon)}{\epsilon}.$

Now, if $\alpha > \frac{\sqrt{\ell}}{32 \rr \log(8\rho/\epsilon)}$ then set $R = \bot$ and again observe that the container function \cref{alg:Container} outputs $C_{|F|}$ on input $(F,R)$ and satisfies the conclusions of the lemma for set $J.$
Otherwise, let $t$ be the step of the procedure that achieves \cref{eqn:worstRatio} described above.
Using that $\rr > \frac{\sqrt{\ell}}{16 \log(8\rho/\epsilon)}$ (by \cref{lem:weakContainment}), apply \cref{lem:revision} with $C = C_{|F|}$ and $D=C_{t-1}$ to find that there exists a vertex $v \in J$ and $\gamma \geq \frac{\sqrt{\ell}}{16 \rr \log(8\rho/\epsilon)}$ such that $|C_{|F|} \setminus (C_{t-1})_{\uparrow v}| \leq (1-\gamma)\rho n$ and $|J \cap (C_{t-1})_{\uparrow v}| \leq \frac{16 \gamma \log^2(8\rho/\epsilon)}{\sqrt{\ell}} |J|.$
Hence, the container $C_{|F|} \setminus (C_{t-1})_{\uparrow v}$ satisfies the lemma conclusions for $J$ because 
\begin{align*}
    |(C_{|F|} \setminus (C_{t-1})_{\uparrow v}) \cap J| 
    & \geq \left(1-\frac{\log(2\rho/\epsilon)}{\rr}-\frac{16\alpha \log(8\rho/\epsilon)}{\sqrt{\ell}}-\frac{16 \gamma \log^2(8\rho/\epsilon)}{\sqrt{\ell}}\right)|J| \\
    & \geq \left(1-\frac{32 \gamma \log^2(8\rho/\epsilon)}{\sqrt{\ell}}\right)|J|
\end{align*}
and $|F| \leq \frac{16 \gamma \rho^2 \log^2(8\rho/\epsilon)}{\epsilon}.$
Hence, set $R=(t-1,v),$ and observe that the container function \cref{alg:Container} outputs $C_{|F|} \setminus (C_{t-1})_{\uparrow v}$ on input $(F,R)$ as desired.

Construct the collection $\mathcal{F}$ by applying the above process to construct a fingerprint $(F,R)$ for every $J \subseteq V$ for which $G[J]$ has fewer than $\frac{\epsilon}{\ell \rho^2}|J|^2$ edges.
This completes the proof of the lemma.
\end{proof}

\begin{remark}
    \label{rem:logNecessity}
    \cref{lem:GCL} provides a collection of containers for sparse subgraphs $G[J]$ with fewer than $\frac{\epsilon}{\ell \rho^2}|J|^2$ edges for $\ell = \Omega(\log^4(1/\epsilon)),$ and, as seen in \cref{thm:tolerantTester}, this implies a $(\epsilon',\epsilon)$-tolerant tester with $\epsilon' \sim \frac{\epsilon}{\log^4(1/\epsilon)}.$
    It is desirable to have an efficient fully tolerant tester, which is a tester that applies for any $\epsilon' < \epsilon,$ because it implies an efficient distance approximation algorithm.
    Using our approach, this would require a graph container lemma that applies for sparse subgraphs $G[J]$ with fewer than $\frac{\epsilon}{\ell \rho^2}|J|^2$ edges for any $\ell > 1.$
    
    While we do not try to optimize the logarithmic dependence in our results, we note here that there is a simple example that demonstrates the challenge with removing the $\log$ factors completely in our approach.
    In particular, consider a graph $G=(V,E)$ formed by planting a $\rho n/2$ independent set $I$ in the random graph $G(n,\frac{8\epsilon}{\rho^2}).$
    Further, select a set $J \subseteq I$ of size $\frac{\rho n}{100}$ (although the size does not matter) and plant a random subgraph $G(|J|,\frac{\epsilon}{\rho^2}).$
    Finally, for each vertex $v$ in $J,$ select $\frac{8 \epsilon}{\rho^2}n$ random vertices in $V \setminus N(v)$ and add edges between $v$ and those vertices so that the vertices in $J$ have the highest degree in the graph.
    First observe that this graph is $\epsilon$-far from \rhoIndepSet.
    Now, consider running \cref{alg:Fingerprint} on $J$ and $G.$
    With high probability, at each step the vertices in $J$ have the highest degree in the graph, and so the only option is to remove neighbours.
    However, doing this removes at $\frac{16\epsilon}{\rho^2}$ fraction of the container in each step, and so it requires at least roughly $\frac{\rho^2}{16\epsilon}\log(1/\rho)$ steps before the container is of size less than $\rho n.$
    However, in doing so we remove at least roughly $\frac{\epsilon}{\rho^2}$ fraction of $J$ in each step, and so after roughly $\frac{\rho^2}{16\epsilon}\log(1/\rho)$ steps only roughly $\rho^{1/16}$ of $J$ is still contained in the container, which in general may be much smaller than $1/2.$
    In order to guarantee that the final container contains more of $J$ we would need $G[J]$ to have edge density less than roughly $\frac{\epsilon}{\rho^2 \log(1/\rho)}.$
\end{remark}

\section{Counting Sparse Subgraphs in Regular Graphs}
In this section we prove \cref{thm:countingSparseSubgraphs}.
The proof follows a similar counting argument to the result by Sapozhenko \cite{sapozhenko2001number} with \cref{lem:GCL} replacing a standard graph container lemma.

The main idea is that, for an appropriate value of $\epsilon,$ any $d$-regular graph is $\epsilon$-far from having an independent set on just over $\frac{n}{2}$ vertices.
Hence, we can apply \cref{lem:GCL} to find a collection of fingerprints and associated containers such that any sparse subgraph is mostly contained in at least one of these containers.
The remaining vertices in the sparse subgraph may come from outside the container, however since there are not too many of them then we can upper bound the number of sparse subgraphs by counting every possible subset of every container, and multiplying by the number of ways that the small number of remaining vertices can be selected from the rest of the graph.

\thmCountingSparseSubgraphs*

\begin{proof}
    Let $G=(V,E)$ be a $d$-regular graph and let $k \geq c \log^9 n$ for a sufficiently large constant $c.$
    First we claim that for any $\epsilon,$ $G$ is $\epsilon$-far from having an independent set of size $\left(\frac{1}{2}+\frac{\epsilon n}{d}\right)n.$
    To show this, consider any set $S \subset V$ of size $\left(\frac{1}{2}+\frac{\epsilon n}{d}\right)n.$
    Observe that $|E(V\setminus S,V)| \leq d \cdot |V \setminus S| = d \left(\frac{1}{2}-\frac{\epsilon n}{d}\right)n.$
    Since there are at least $\frac{d n}{2}$ edges in the graph, then $|E(S)| \geq \frac{d n}{2} - d\left(\frac{1}{2}-\frac{\epsilon n}{d}\right)n = \epsilon n^2.$

    We will select a specific $\epsilon$ and $\ell$ and apply \cref{lem:GCL} using $\rho = \left(\frac{1}{2}+\frac{\epsilon n}{d}\right)n.$
    To do this we consider two cases.
    First, if $k < d^3$ then we set $\epsilon = c_1\frac{d \log^3(n)}{k^{1/3} n}$ for a sufficiently large constant $c_1.$
    Otherwise, if $k \geq d^3,$ set $\epsilon = c_2\frac{\log^3(n)}{n}$ for a sufficiently large constant $c_2.$
    In either case, let $\ell = \frac{2\epsilon k n}{d}$ so that the number of edges in any subgraph $J \subseteq V$ with edge density at most $\frac{1}{k}\frac{d}{n}$ is at most \[\frac{1}{k}\frac{d}{n}\binom{|J|}{2} \leq \frac{1}{2k}\frac{d}{n} |J|^2 = \frac{\epsilon}{\ell}|J|^2 \leq \frac{\epsilon}{\ell \rho^2} |J|^2.\]
    Hence, by \cref{lem:GCL}, there exists an absolute constant $c_3$ and a collection $\mathcal{F}$ of fingerprints, and a container function $C: \mathcal{F}\rightarrow P(V)$ such that for any $J \subseteq V$ for which $G[J]$ has edge density less than $\frac{1}{k}\frac{d}{n},$ there exists $(F,R) \in \mathcal{F}$ and $\alpha$ such that $|C(F,R)| = (1-\alpha)\rho n,$ $|F| \leq \frac{c_3 \alpha \rho^2 \log^2(n)}{\epsilon}$ and $|J \cap C(F,R)| \geq \left(1-\frac{c_3\alpha\log^2(n)}{\sqrt{\ell}}\right)|J|,$ where we use that $1/\epsilon \leq n$ for either choice of $\epsilon$ in upper bounding the logarithms.

    Since $\alpha \leq 1$ and $|J| \leq n,$ the third conclusion can be simplified to say that for any such sparse subgraph $G[J],$ there exists $(F,R) \in \mathcal{F}$ such that all but at most $\frac{c_3 \log^2n}{\sqrt{\ell}} n$ of $J$ is contained in $C(F,R).$
    Hence, the number of such sparse subgraphs is at most
    \begin{equation}
        \label{eqn:sparseSubgraphsUpperBound1}
        \sum_{(F,R) \in \mathcal{F}} 2^{|C(F,R)|} \cdot \binom{n}{\frac{c_3 n\log^2n}{\sqrt{\ell}}}.
    \end{equation}
    Let $f_{\max}=\frac{c_3 \rho^2 \log^2(n)}{\epsilon},$ and split the sum up into fingerprints $(F,R)$ for which $|F|=f$ for ${ f=1,\dots,f_{\max} }.$
    For such a fingerprint $(F,R)$ with $|F|=f,$ observe that \[|C(F,R)| \leq \left(1-\frac{\epsilon f}{c_3 \rho^2 \log^2n}\right) \rho n = \left(1-\frac{\epsilon f}{c_3 \rho^2 \log^2n}\right)\left(\frac{1}{2}+\frac{\epsilon n}{d}\right)n \leq \frac{n}{2} + \frac{\epsilon n^2}{d} - \frac{\epsilon f n}{2c_3 \log^2n}.\]
    Also observe that there are at most $(2n)^f \cdot n \cdot f \leq (2n)^{f+2}$ possible $(F,R)$ with $|F|=f.$
    Substituting into \cref{eqn:sparseSubgraphsUpperBound1}, we find that the number of $J \subseteq V$ for which $G[J]$ has edge density less than $\frac{1}{k}\frac{d}{n}$ is at most
    \begin{equation}
        \label{eqn:sparseSubgraphsUpperBound2}
        \sum_{f=1}^{f_{\max}} (2n)^{f+2} ~ 2^{\frac{n}{2} + \frac{\epsilon n^2}{d} - \frac{\epsilon f n}{2c_3 \log^2n}} ~ \binom{n}{\frac{c_3 n\log^2 n}{\sqrt{\ell}}} \leq \sum_{f=1}^{f_{\max}} 2^{\frac{n}{2} + \frac{\epsilon n^2}{d} - \frac{\epsilon f n}{2c_3 \log^2n} + (f+2)\log(2n) + \frac{c_3 n \log^3n}{\sqrt{\ell}}},
    \end{equation}
     where we use the upper bound $\binom{n}{\frac{c_3 n \log^2n}{\sqrt{\ell}}} \leq n^{\frac{c_3 n \log^2n}{\sqrt{\ell}}}.$

    Now, we consider the $\epsilon$ selections for the two cases.
    First, if $k <d^3$ then use that $\epsilon = c_1\frac{d \log^3(n)}{k^{1/3} n}$ and $\ell=\frac{2\epsilon k n}{d}=2c_1 k^{2/3} \log^3n,$ and observe that $\frac{\epsilon f n}{2c_3 \log^2n} \geq (f+2)\log(2n)$ as long as $c_1$ is sufficiently large.
    Substituting into \cref{eqn:sparseSubgraphsUpperBound2} we find
    \[\sum_{f=1}^{f_{\max}} 2^{\frac{n}{2} + \frac{c_1 n\log^3(n)}{k^{1/3}} + \frac{c_3 n \log^{3/2}n}{\sqrt{2c_1} k^{1/3}}}= \sum_{f=1}^{f_{\max}} 2^{\frac{n}{2}\left(1+O\left(\frac{\log^3n}{k^{1/3}}\right)\right)}=2^{\frac{n}{2}\left(1+O\left(\frac{\log^3n}{k^{1/3}}\right)\right)},\]
    where the last equality uses that $f_{\max} \leq n.$

    Next, in the case that $k \geq d^3,$ use that $\epsilon = c_2\frac{\log^3(n)}{n}$ and $\ell=\frac{2\epsilon k n}{d}=\frac{2c_2 k \log^3n}{d} \geq 2c_2 k^{2/3} \log^3 n,$ and again observe that $\frac{\epsilon f n}{2c_3 \log^2n} \geq (f+2)\log(2n)$ as long as $c_2$ is sufficiently large.
    Substituting into \cref{eqn:sparseSubgraphsUpperBound2} we find
    \[\sum_{f=1}^{f_{\max}} 2^{\frac{n}{2} + \frac{c_2 n\log^3(n)}{d} + \frac{c_3 n \log^{3/2}n}{\sqrt{2c_2} k^{1/3}}}= \sum_{f=1}^{f_{\max}} 2^{\frac{n}{2}\left(1+O\left(\frac{\log^3n}{d}\right)+ O\left(\frac{\log^{3/2}n}{k^{1/3}}\right)\right)}=2^{\frac{n}{2}\left(1+O\left(\frac{\log^3n}{d}\right)+O\left(\frac{\log^{3/2}n}{k^{1/3}}\right)\right)},\]
    where the last equality again uses that $f_{\max} \leq n.$

    Combining, we have that the number of $J \subseteq V$ for which $G[J]$ has edge density less than $\frac{1}{k}\frac{d}{n}$ is at most
    \[2^{\frac{n}{2}\left(1+O\left(\frac{\log^3n}{d}\right)+O\left(\frac{\log^3n}{k^{1/3}}\right)\right)}.\]
\end{proof}

\begin{remark}
    \label{rem:countSparseSubgraphsExample}
    We can also show that the graph formed by $\frac{n}{2d}$ disjoint copies of $K_{d,d},$ which is $d$-regular, has at least roughly $2^{\frac{n}{2}\left(1+\frac{1}{2d}+\Omega\left(\frac{\log(k)}{k}\right)\right)}$ induced subgraphs with edge density at most $\frac{1}{k} \frac{d}{n}.$
    To show this, consider any subset $J \subseteq V$ formed by selecting a half from each $K_{d,d},$ and selecting at least $d/4$ vertices from that half and fewer than $\frac{d}{32k}$ vertices from the other half.
    Observe that the $\frac{d}{32k}$ vertices each have degree at most $d,$ and so any such subgraph $G[J]$ has at most $\frac{n}{2d} \frac{d}{32k} \cdot d = \frac{dn}{64k}$ edges.
    Since there are at least $n/8$ vertices in any such $J$, then the edge density is at most $\frac{dn}{k}.$

    Now, we count the number of possible sets $J.$
    There are $2^{\frac{n}{2d}}$ ways to select a half for each of the $\frac{n}{2d}$ copies of $K_{d,d},$ there are at least $2^d - \binom{n}{n/4}$ ways to select the larger half of each $K_{d,d}$, and at least $\binom{d}{\frac{d}{32k}}$ ways to select the smaller half.
    Hence, the number of such subgraphs is at least \[2^{\frac{n}{2d}} \cdot \left(2^d - \binom{d}{d/4}\right)^{\frac{n}{2d}} \cdot \binom{d}{\frac{d}{32k}}^{\frac{n}{2d}}.\]
    Use that $\binom{d}{d/4} \leq 2^{\frac{d}{4}\log(4e)} \leq 2^{7d/8}$ and $\binom{d}{\frac{d}{32k}} \geq (32k)^{\frac{d}{32k}} = 2^{\frac{d}{32k}\log(32k)}$ to find that the number of subgraphs with edge density less than $\frac{d}{kn}$ is at least    
    \[2^{\frac{n}{2}\left(1+\frac{1-2^{-\Omega(d)}}{d}+\Omega\left(\frac{\log(k)}{k}\right)\right)}.\]
\end{remark}

\section*{Acknowledgements}
Thank you to Eric Blais for many useful discussions throughout this project and for feedback on early versions of the paper.
Thank you to Asaf Shapira for useful comments regarding prior work on testing, and tolerantly testing, graph partition properties.
We also thank Robert Andrews and Janani Sundaresan for reading and giving feedback on the paper.
Finally, we thank Aadila Ali Sabry for the interesting discussions related to this work during the Directed Reading Program at Waterloo.

\bibliographystyle{alpha}
\bibliography{references}

\end{document}